\documentclass[11pt]{llncs}

 \usepackage{epsfig}
 \usepackage{enumitem}
 \usepackage{amsfonts}
 \usepackage{amssymb}
 \usepackage[center,scriptsize]{subfigure}
 \usepackage{multirow}

\usepackage{array}
\newcolumntype{H}{>{\setbox0=\hbox\bgroup}c<{\egroup}@{}}

\usepackage[margin=1in]{geometry}

\newtheorem{observation}[theorem]{Observation}

\usepackage[noend]{algorithmic}
\usepackage{algorithm,verbatim}

\usepackage{tikz}
\usetikzlibrary{decorations.pathmorphing}

\newcommand{\REM}[1]{}

\newcommand {\ffullyupdate}{{\sc ff-update } }
\newcommand {\ffullyupdateend}{{\sc ff-update}}
\newcommand {\ffullydynamic}{{\sc ffd} }
\newcommand {\ffullydynamicend}{{\sc ffd}}
\newcommand {\ffullycleanup}{{\sc ff-cleanup} }
\newcommand {\ffullycleanupend}{{\sc ff-cleanup}}
\newcommand {\ffullyfixup}{{\sc ff-fixup} }
\newcommand {\ffullyfixupend}{{\sc ff-fixup}}
\newcommand {\ffullypopulate}{{\sc ff-populate-heap}}
\newcommand {\ffullygetnew}{{\sc ff-new-paths}}
\newcommand {\ffullyccenters}{{\sc ff-cleanup-centers }}

\newcommand {\ffullyccentersend}{{\sc ff-cleanup-centers}}
\newcommand {\ffullyfcentersend}{{\sc ff-fixup-centers}}
\newcommand {\LHT}{LHT }
\newcommand {\LHTe}{LHT}
\newcommand {\HT}{HT }

\newcommand {\LHPe}{LHP}

\newcommand {\NODES}{\mathcal{N}}

\newcommand {\CA}{C}
\newcommand {\DMs}{DL}
\newcommand {\levelgraph}{\Gamma}

\newcommand {\LN}{LN}
\newcommand {\RN}{RN}

\newcommand {\system}{tuple-system }
\newcommand {\systemend}{tuple-system}

\newcommand {\weight} {{\bf {w}}}

\newcommand {\dict} {dict}
\newcommand {\MT} {Marked-Tuples }
\newcommand {\MTend} {Marked-Tuples}
\newcommand {\vstar} {{\nu^*}}

\newcommand{\vone}{\vspace{.1in}}
\newcommand{\vhalf}{\vspace{.05in}}

\newcommand{\DI}{\texttt{DI }}
\newcommand{\PR}{\texttt{PR }}
\newcommand{\NPR}{\texttt{NPRdec }}
\newcommand{\Tho}{\texttt{Thorup }}
\newcommand{\FFD}{{\sc ffd }}
\newcommand{\DIe}{\texttt{DI}}
\newcommand{\PRe}{\texttt{PR}}
\newcommand{\NPRe}{\texttt{NPRdec}}
\newcommand{\Thoe}{\texttt{Thorup}}
\newcommand{\FFDe}{{\sc ffd}}
\newcommand{\BDAGe}{{\sc build-dag}}

\usepackage{url}
\urldef{\mailid}\path|{cavia,vlr}@cs.utexas.edu|

\newcommand{\hide}[1]{}

\newcommand{\Xomit}[1]{}

\begin{document}
\title {A Faster Algorithm for Fully Dynamic Betweenness Centrality\thanks{This work was supported in part by NSF grant CCF-1320675. Authors' affiliation: Dept. of Computer Science, University of Texas, Austin, TX 78712.}}
\author {Matteo Pontecorvi and Vijaya Ramachandran}
\institute{
University of Texas at Austin, USA \\ \mailid
}

 \maketitle

\pagenumbering{arabic}

\begin{abstract}
We present a new fully dynamic algorithm for maintaining betweenness centrality (BC) of vertices in a directed
graph $G=(V,E)$ with positive edge weights. BC is a widely used parameter
in the analysis of large complex networks.
We achieve an amortized $O(\vstar^2 \cdot \log^2 n)$ time per update, where  $n = |V| $ and $\vstar$ bounds  the number of distinct edges that lie on shortest paths through any single vertex. This result improves on  the amortized bound for fully dynamic BC in \cite{PR14,PR15i} by a logarithmic factor. Our algorithm  uses  new data structures and techniques that are extensions of the method in the fully dynamic 
algorithm in Thorup \cite{Thorup04} for APSP in graphs with unique shortest paths.
For graphs with $\nu^* = O(n)$, our algorithm matches the fully dynamic APSP
bound in Thorup~\cite{Thorup04}, which holds for graphs with $\nu^* = n-1$, since it assumes unique shortest paths.
\end{abstract} 

\section{Introduction}

Betweenness centrality (BC) is a widely-used measure
in the analysis of large complex networks, and is defined as follows.
Given a directed graph $G=(V,E)$ with $|V|=n$, $|E|=m$ and positive edge weights, let 
$\sigma_{xy}$ denote the number of shortest paths (SPs)
from $x$ to $y$ in $G$, and
$\sigma_{xy}(v)$ the number of SPs from $x$ to $y$ in $G$ that pass through $v$, for each pair $x,y\in V$. 
Then, $BC(v) = \sum_{s \neq v, t \neq v} \frac{\sigma_{st}(v)}{\sigma_{st}}$.
As in \cite{Brandes01}, we assume positive edge weights to avoid the case when cycles of 0 weight are present in the graph. 

The measure $BC(v)$  is often used as an index that  determines the relative importance of
$v$ in $G$, and is computed for all $v\in V$.
Some applications of
BC include analyzing social interaction networks \cite{KA12},
identifying lethality in biological networks \cite{PCW05},
and identifying key actors in terrorist networks \cite{Coffman,Krebs02}.
In the static case,
the widely used algorithm by Brandes~\cite{Brandes01} runs in 
$O(mn + n^2 \log n)$ on weighted graphs.
Several approximation algorithms are available: \cite{BaderKMM07,Riondato} for static computation and, recently, \cite{Bergamini1,Bergamini2} for dynamic computation.
Heuristics for dynamic betweenness centrality with good experimental
performance are given in~\cite{GreenMB12,Lee12,SinghGIS13},
but none provably improve on Brandes.
The only earlier exact  dynamic BC algorithms that provably improve on
Brandes on
some classes of graphs are the recent separate incremental and decremental\footnote{Incremental/decremental 
	refer to the insertion/deletion of a
	vertex or edge; the corresponding weight changes that apply are 
	weight decreases/increases, respectively.}
algorithms in~\cite{NPR14,NPR14b}. Recently, we give in \cite{PR14} (see also \cite{PR15i}) the first fully dynamic algorithm for BC (the \PR method) that is provably faster than  
Brandes for the class of dense graphs 
(where $m$ is close to $n^2$)
with succinct single-source SP dags.
Table~\ref{table1} contains a summary of these results.

In this paper, we present an improved algorithm for computing fully dynamic exact betweenness centrality: 
our algorithm \FFD 
improves over \PR by a logarithmic factor
using data structures and technique that are considerably more complicated.
Our method is a generalization of
Thorup~\cite{Thorup04} (the \Tho method) which computes fully dynamic APSP
 for graphs with a unique SP for every vertex pair; however, a key step in BC algorithms is computig  \emph{all} SPs for each pair of vertices (\emph{all pairs 
 \underline{all} shortest paths --- AP\underline{A}SP}).    We develop a faster fully dynamic algorithm for APASP, which in turn 
  leads to a faster fully dynamic BC algorithm than \PRe.

\vone
\noindent
{\bf Our Results.} 
Let $\vstar$ be the maximum
number of distinct edges that lie on shortest paths through any given vertex in $G$; 
for convenience we assume $\nu^* = \Omega (n)$.

\begin{theorem}\label{th:main}
Let $\Sigma$ be a sequence of $\Omega (n)$ 
fully dynamic vertex updates on a directed $n$-node graph $G=(V,E)$ with positive
edge weights. Let $\nu^*$ bound the number of distinct edges that lie on shortest paths 
through any single vertex in any of the updated graphs or their vertex induced subgraphs. 
Then, all BC scores (and APASP) can be
maintained in amortized time $O( \vstar^2 \cdot \log^2 {n})$ per update with algorithm \ffullydynamicend.
\end{theorem}

Similar to the \PR algorithm, our new algorithm \FFD is provably faster than 
Brandes on dense graphs with succinct single-source SP dags. It also matches the \Tho bound for APSP when $\nu^* = O(n)$.

Our techniques rely on recomputing BC scores using certain data structures related to shortest paths extensions (see Section \ref{sec:fdbc}). These are generalizations of structures introduced by Demetrescu and Italiano  \cite{DI04} for fully dynamic APSP (the \DI method)
and \Thoe, where only one SP is maintained for each pair of vertices. Our generalizations build on the
 \system introduced in~\cite{NPR14b} (the \NPR method) for decremental APASP (see Section \ref{sec:bkg}), which is a method to
 succinctly  represent all of the multiple SPs for every pair of vertices.  Our algorithm also
 builds on the fully dynamic APASP (and BC) algorithm in  \PRe, which runs in amortized $O(\vstar^2 \cdot \log^3n)$ cost per update (see Section \ref{sec:bkg}). Finally, one of the main challenges we address in our current result is to  generalize the `level graphs' of \Tho to the case when different SPs for a given vertex pair can be distributed across multiple levels.

\begin{table}
	\begin{center}
		{\scriptsize
			\begin{tabular}{c | c | H | c | c | c | c | c  }
				\hline
				\textbf{Paper} & \textbf{Year} & \textbf{Space} & \textbf{Time} & \textbf{Weights} & \textbf{Update Type} & \textbf{DR/UN} & \textbf{Result} \\ 
				\hline
				\hline
				Brandes~\cite{Brandes01} & 2001 & $O(m+n)$ & $O(mn)$ & NO & Static Alg. & Both & Exact \\
				Brandes~\cite{Brandes01} & 2001 & $O(m+n)$ & $O(mn+n^2\log n)$ & YES & Static Alg. & Both & Exact \\
				Geisberger~et~al.~\cite{Geisb} & 2007 & -- & Heuristic & YES & Static Alg. & Both & Approx. \\
				Riondato~et~al.~\cite{Riondato} & 2014 & $O(n^2)$ & depends on $\epsilon$ & YES & Static Alg. & Both & $\epsilon$-Approx. \\
				\hline
				\multicolumn{8}{c}{\textbf{Semi Dynamic}} \\
				\hline
				Green~et~al.~\cite{GreenMB12} & 2012 & $O(n^2+mn)$ & $O(mn)$ & NO & Edge Inc. & Both & Exact \\ 
				Kas~et~al.~\cite{KasWCC13} & 2013 & $O(n^2+mn)$ & Heuristic & YES & Edge Inc. & Both & Exact \\ 
				NPR~\cite{NPR14} & 2014 & $O(\vstar \cdot n)$ & $O(\vstar \cdot n)$ & YES & Vertex Inc. & Both & Exact \\ 
				NPRdec~\cite{NPR14b} & 2014 & $O(m^* \cdot \vstar)$ & $O(\vstar^2 \cdot \log n)$ & YES & Vertex Dec. & Both & Exact \\ 
				Bergamini~et~al.~\cite{Bergamini1} & 2015 & depends on $\epsilon$ & depends on $\epsilon$ & YES & Batch (edges) Inc. & Both & $\epsilon$-Approx. \\ 
				\hline 
				\multicolumn{8}{c}{\textbf{Fully Dynamic}} \\
				\hline
				Lee~et~al.~\cite{Lee12} & 2012 & $O(n^2+m)$ & Heuristic & NO & Edge Update & UN & Exact \\
				Singh~et~al.~\cite{SinghGIS13} & 2013 & -- & Heuristic & NO & Vertex Update & UN & Exact \\ 
				Kourtellis+~\cite{Kourtellis} & 2014 & $O(n^2)$ & $O(mn)$ & NO & Edge Update & Both & Exact \\
				Bergamini~et~al.~\cite{Bergamini2} & 2015 & depends on $\epsilon$ & depends on $\epsilon$ & YES & Batch (edges) & UN & $\epsilon$-Approx. \\
				PR~\cite{PR14,PR15i}  & 2015 & $\tilde{O}(m \cdot \vstar)$ & $O(\vstar^2 \cdot \log^3 n)$ & YES & Vertex Update & Both & Exact \\ 
				\textbf{This paper (\FFDe)} & 2015 & $\tilde{O}(m \cdot \vstar)$ & $O(\vstar^2 \cdot \log^2 n)$ & YES & Vertex Update & Both & Exact \\ 
				\hline 
			\end{tabular} 
		}
		\caption{Related results (DR stands for Directed and UN for Undirected)}
	\end{center}
	\label{table1}
\end{table}

\noindent
{\bf Discussion of the parameters $m^*$ and $\vstar$.}
Let $m^*$ be the number of distinct edges in $G$ that lie on shortest paths; $\nu^*$, defined above, is the maximum number of distinct
edges on any single source SP dag. Clearly, $\nu^* \leq m^* \leq m$.\\
- \underline{$m^* \textrm{ vs } m$}: In many cases,  $m^* \ll m$:
	as noted in \cite{KKP93}, in a complete graph ($m=\Theta(n^2$)) where edge weights are chosen from
	a large class of probability distributions, $m^* = O(n \log{n})$
	with high probability. 
	
	\vhalf
	\noindent
- \underline{$\vstar \textrm{ vs } m^*$}: Clearly, $\vstar =O(n)$ in any graph with only a constant number of SPs between every
	pair of vertices. These graphs are called $k$-geodetic \cite{kgeo} (when at most $k$ SPs exists between two nodes), and are well studied in graph theory \cite{bigeo,bandelt,mulder}.
In fact $\vstar = O(n)$ even in some graphs that have an exponential number of SPs between some pairs of vertices.

    In contrast, $m^*$ can be $\Theta (n^2)$  even in some graphs with unique SPs, for example the complete unweighted
    graph $K_n$.
    
    Another type of graph with  $\vstar \ll m^*$ is one with large clusters of nodes
    (e.g., as described by the \emph{planted $\ell$-partition model} \cite{Condon,clusters}).
    Consider a graph $H$ with $k$ clusters of size  $n/k$ (for some constant $k \geq 1$) with
    $\delta < w(e) \leq 2\delta$, for some constant $\delta>0$,  for each edge $e$ in a cluster;
    between the clusters is a sparse
    interconnect. Then
    $m^* = \Omega(n^2)$ but $\vstar= O(n)$. (The  connections between clusters can be arbitrary, thus  BC scores are non-trivial to compute.)
    For all of the above classes of graphs, our BC algorithm
    will run in amortized $\tilde O(n^2)$ 
    time per update ($\tilde O$ hides polylog factors). More generally we have:
    
    \begin{theorem}\label{th:main2}
    	Let $\Sigma$ be a sequence of $\Omega (n)$  updates on 
    	graphs with
    	$O(n)$ distinct edges on shortest paths through any single vertex
    	in any vertex-induced subgraph.
    	Then, all BC scores (and APASP) can be
    	maintained in amortized time $O( n^2 \cdot \log^2 {n})$ per update.
    \end{theorem}

Our algorithm uses $\widetilde{O}(m \cdot \vstar)$ space, extending the $\widetilde{O}(mn)$ 
result
in \DI for APSP. Brandes uses only linear space, but all
known
dynamic algorithms require at least 
$\Omega (n^2)$ space.
	
\vhalf
\noindent
{\bf Overview of the Paper.} 
In Section \ref{sec:fdbc} we describe how to obtain a fully dynamic BC algorithm using a fully dynamic  APASP algorithm. The remaining
sections in the paper are devoted to developing our improved fully dynamic APASP algorithm.
In Section \ref{sec:bkg} we review the \NPRe, \PR and \Tho methods. 
In Section \ref{subsec:fully-impl} we introduce the level \system framework for APASP, with particular reference to  the new data structures specifically developed for our result.
In Section~\ref{sec:algo} we  present our algorithm \FFDe, and we describe its main components in detail.
Section \ref{sec:ffeatures} describes two important challenges that arise when generalizing \Tho to APASP setting; addressing these two challenges is crucial to the correctness and efficiency of our algorithm. 
Finally, in Section \ref{sec:proof} we establish the  amortized time bound of $O({\nu^*}^2 \cdot \log^2 n)$ for \FFD and its correctness. 
\section{Fully Dynamic Betweenness Centrality} \label{sec:fdbc}

The static Brandes algorithm~\cite{Brandes01} computes BC scores in a two phase process. The first phase computes the SP out-dag
for every source through $n$ applications of Dijkstra's algorithm.  The second phase
uses an `accumulation' technique that computes all BC scores  using these SP dags in  $O(n \cdot \nu^*)$ time.

In our fully dynamic algorithm, we will leave the second phase unchanged. For the first phase, we will use the approach 
in the incremental BC algorithm in~\cite{NPR14}, which maintains the SP dags using a very simple and efficient 
incremental algorithm. For decremental and fully dynamic updates the method is more involved, and dynamic APASP is at the heart of
maintaining the SP dags.  Neither the decremental nor our new fully dynamic APASP algorithms  maintain the SP dags explicitly, instead they maintain data structures to update a  collection of {\it tuples} (see Section~\ref{sec:bkg}). We now
describe a very simple method to construct the SP dags from these data structures; this step is not addressed in the decremental
\NPR algorithm and it is only sketched in the fully dynamic \PR algorithm.

For every vertex pair $x,y$, the following sets $R^*(x,y)$, $L^*(x,y)$ (introduced in \DIe) are maintained 
in \NPRe, and in both of our fully dynamic algorithms:

\vhalf
\noindent
- $R^*(x,y)$ contains all nodes $y'$ such that every shortest path $x\rightsquigarrow y$ in $G$ can be extended with the edge $(y,y')$ to generate another shortest path $x\rightsquigarrow y \rightarrow y'$. \\
- $L^*(x,y)$ contains all nodes $x'$ such that every shortest path $x\rightsquigarrow y$ in $G$ can be extended with the edge $(x',x)$ to generate another shortest path $x' \rightarrow x\rightsquigarrow y$.\\
These sets allow us to construct the SP dag for each source $s$ using the following simple algorithm \BDAGe.

\begin{algorithm}
	\begin{algorithmic}[1]
		\FOR {each $t \in V$}
		\FOR {each $u \in R^*(s,t)$}
		\STATE {\bf if} {$D(s,t)+\weight(t,u) = D(s,u)$} {\bf then} add the edge $(t,u)$ to dag($s$) \label{buildDAG:check}
		\ENDFOR
		\ENDFOR
	\end{algorithmic}
	\caption{\BDAGe($G, s, \weight, D)$ \scriptsize ($\weight$ is the weight function; $D$ is the distance matrix)}
	\label{algo:buildDAG}
\end{algorithm}

In our fully dynamic algorithm $R^*$ and $L^*$ will be supersets of the exact collections of nodes defined above, but the check in 
Step~\ref{buildDAG:check} will ensure that only the correct SP dag edges are included. The combined sizes of these $R^*$ and
$L^*$ sets is $O(n \cdot \nu^* \cdot \log n)$ in our \FFD algorithm, hence the
amortized time bound for the overall fully dynamic BC algorithm is dominated by the time bound for fully dynamic APASP.
In the rest of this paper, we will present our fully dynamic APASP algorithm \FFDe.
\section{Background} \label{sec:bkg}
In this section we review prior work upon which we build our results.
For each, we highlight the inherited notation we use and the main ideas we extend. 

\subsection{The \NPR Decremental APASP Algorithm \cite{NPR14b}}
\NPR generalizes the decremental APSP algorithm in \DI \cite{DI04} to obtain a decremental algorithm for APASP and BC.
For the decremental APSP algorithm
\DI develops a novel method to maintain (unique) shortest paths \cite{DI04}. \DI defines an \emph{LSP} as a path where every proper subpath is a shortest path in the graph. By efficiently maintaining all LSPs  after each update, \DI presents an efficient
 decremental APSP, which is then extended to a fully dynamic APSP algorithm with additional tools.
This provides a fully dynamic algorithm for APSP that runs in $O(n^2\log n)$ amortized time per update.  
 \NPR extends this result to APASP by introducing the \emph{\systemend} to replace the need to maintain every SP and LSP in \DIe.
 We now briefly review this system, referring the reader to~\cite{NPR14b} for more details. Let 
$d(x,y)$ denote the shortest path length from $x$ to $y$.

A {\it tuple}, $\tau = (xa, by)$, represents a set of paths in $G$, all with the same weight, and all
of which  use
the same first edge $(x,a)$ and the same last edge $(b,y)$. 
If the paths in $\tau$ are LSPs, then $\tau$ is
an LST (locally shortest tuple), and the
weight of every path
in  $\tau$ is $\weight(x, a)$ + $d(a, b) + \weight(b,y)$. 
In addition, if $d(x, y) = \weight(x, a) + d(a, b) + \weight(b, y)$, then 
$\tau$ is a {\it shortest path tuple (ST)}. 

A {\it triple} $\gamma=(\tau, wt, count)$, represents the tuple $\tau=(xa,by)$ that contains
$count$ number of paths from $x$ to $y$, each with weight
$wt$. 
\NPR  uses triples to succinctly store 
all LSPs and SPs
for each vertex pair
 in $G$.
 
 \noindent
{\it Left Tuple and Right Tuple.}
A left tuple (or $\ell$-tuple), 
$\tau_{\ell} = (xa, y)$,
represents the
set of LSPs from $x$ to $y$, all of which use 
the same first edge $(x,a)$.
A right tuple ($r$-tuple) $\tau_r = (x, by)$ is defined analogously.
For a shortest path $r$-tuple $\tau_r = (x, by)$,  $L({\tau_r})$ is
the set of vertices which can be used as pre-extensions to create  LSTs
in~$G$, and
 for a shortest path $\ell$-tuple $\tau_{\ell} = (xa, y)$,
$R(\tau_{\ell})$ is the set of vertices which can be used as post-extensions to
create LSTs in $G$.

\NPR maintains several other sets such as  $P^*$ and $P$ for each vertex pair. Since our algorithm also maintains generalizations of
these sets, we will discuss them when we present the data structures used by our algorithm in Section~\ref{subsec:fully-impl}.

Similar to \DIe, the \NPR algorithm first deletes all the paths containing the updated node using a procedure {\sc cleanup}, and then updates the \system to maintain all  shortest paths in the graph using procedure {\sc fixup}. 
The main difference, in terms of data structures, is the use of tuples to collect paths that share the same first edge, last edge and weight.

\subsection{The \PR Fully Dynamic APASP Algorithm \cite{PR14}}

In our recent fully dynamic \PR algorithm for APASP \cite{PR14}, 
we build on the \system introduced in \NPRe. \PR also incorporates several elements in the \DI method of extending their decremental 
APSP algorithm to fully dynamic (though some key elements in \PR are significantly different from \DIe). 

One difference between \NPR and \PR is the introduction of HTs and LHTs; these are extension of historical and locally historical paths in \DI to tuples and triples defined as follows (THTs and TLHTs are not specifically used in this paper, however we include their definitions for completeness):
\begin{paragraph}{\bf HT, THT, LHT, and TLHT.}
	Let $\tau$ be a tuple in the \system at time $t$. Let $ t' \le t$ denote the most recent step at which
	a vertex on a path in $\tau$ was updated. Then $\tau$  is  a {\it  historical tuple (HT) at time $t$} if
	$\tau$ was an ST-tuple at least once in the interval $[t', t]$; $\tau$ is a
	{\it true HT (THT) at time $t$} if it is not an ST 
	in the current graph.
	A tuple $\tau$ is a {\it  locally historical tuple (LHT) at time $t$}
	if either it only
	contains a single vertex or every proper sub-path in it is an HT
	at time $t$; a tuple $\tau$ is a
	{\it true LHT (TLHT) at time $t$} if it is not an LST in the current graph.
\end{paragraph}

\vone 
\noindent
Similar to \DI for unique SPs, \PR forms LHTs and TLHTs in its {\sc fully-fixup} procedure (which adapts the {\sc fixup} procedure in \NPR to the fully dynamic case)
by combining compatible pairs of HTs.
An important method introduced in \DI for efficiency in fully dynamic APSP is  the notion of a `dummy update' sequence.
Extending this method to  an efficient algorithm for APASP does not appear to be feasible, so in \PRe, a new dummy update sequence (that uses elements in a different method  by \Tho \cite{Thorup04}) is introduced. \PR then defines the \emph{Prior Deletion Graphs} (PDGs) (that are related to the {\it level graphs} maintained in \Tho -- see below) to study the complexity of the \PR algorithm. In our algorithm \FFD we will use  graphs similar to the PDGs in \PRe; these are described in
 Section~\ref{subsec:fully-impl}.

\subsection{The \Tho Fully Dynamic APSP Algorithm \cite{Thorup04}}
In \cite{Thorup04}, Thorup improves by a logarithmic factor over \DI (for unique shortest paths) by using a \emph{level system} of decremental only graphs. 
The shortest paths and locally shortest paths are generated level by level leading to a different complexity analysis from \DIe. When a node is removed from the current graph, it is also removed from every older level graph that contains it. 
The implementation of the \Tho APSP algorithm is not fully specified in \cite{Thorup04}. For our \FFD algorithm, we present generalizations of the data structures sketched in \Tho (see Section \ref{subsec:fully-impl} for a summary of these data structures).     

\section{Data Structures for Algorithm  \FFDe} \label{subsec:fully-impl}
Our  algorithm  \FFD requires several data structures. Some of these are already present in \NPRe, \PR and \Thoe, while others are newly defined or generalized from earlier ones. We will use components from \PR such as the abstract representation of the level system using PDGs (see Section \ref{subsec:fully-impl}) and the flag bit $\beta$ for a triple, the 
 \MT scheme introduced in \NPR (see \cite{NPR14b,PR14} for more details), and the maintenance of level  graphs from \Thoe.

In this section, we describe all data structures used by our algorithm. 
In Table \ref{table:dsA} we summarize the structures we use, including those inherited from \cite{NPR14b,PR14,Thorup04}.
The new components we introduce in this paper to achieve efficiency for fully dynamic APASP, are described in section \ref{sec:newfeat} and listed in Table \ref{table:dsA}, Part D.

\subsection{A Level System for Centered Tuples} \label{sec:tuple-system}
Algorithm \FFD  uses the PDGs defined in \PR as real data structures similar to \Tho for APSP. This is done in order to generate a smaller superset of LSTs than \PRe, and this is the key to achieving the improved efficiency. Here we describe the level system and the data structures we use in \FFDe, with special attention to
the new elements we introduce.

As in~\cite{DI04,PR14,Thorup04} we build up the \system for the initial $n$-node graph $G=(V,E)$ with $n$ insert updates (starting
with the empty graph), and we then
perform $n$ updates according to the update sequence $\Sigma$. After $2n$ updates, we
reset all data structures and start afresh.

Our level system is a generalization of \Tho to fully dynamic APASP.
For an update at step $t$, let $k$ be the position of the least significant bit with value 1 in the binary representation of $t$. Then
the $t$-th update activates the level $k$, and deactivates all levels $j<k$ by folding these levels into level $k$.
These levels are considered implicitly in \PRe, and using the same notation, we will say that $time(k)=t$, and $level(t)=k$; moreover $G_t$ indicates the graph after the $t$-th update.
Note that the largest level created before we start afresh is $r= \log 2n$.

\paragraph{\bf Centering vertices and tuples/triples.} As in \Thoe, each \emph{vertex $v$ is centered} in level $k=level(t)$, where $t$ is the most recent step in which $v$ was updated.
A path $p$ in a tuple is centered in level $k'=level(t')$, where $t'$ is the most recent step in which $p$ entered the tuple system (within some tuple) or was modified by a vertex update. 
Hence, in contrast to \Thoe, a triple can represent paths centered in different levels. Thus,
for a triple $\gamma = ((xa, by), wt, count)$ we maintain an array $\CA_{\gamma}$ where 
$$\CA_{\gamma}[i] = \textrm{number of paths represented by $\gamma$ that are all centered in level $i$}$$

\noindent
It follows that $\sum_i{\CA_{\gamma}[i]}=count$. 
The \emph{level center of the triple $\gamma$} is the smallest (i.e., most recent level) $i$ such that $\CA_{\gamma}[i] \neq 0$.

\paragraph{\bf Level graphs (PDGs).}
The \PR algorithm defines PDGs as follow: Let $t$ be the current update step,  let $t' <t$, and let $W$ be the set of vertices that are updated in the
interval of steps $[t'+1, t]$. The \emph{prior deletion graph (PDG)} $\Gamma_{t'}$ is the induced subgraph of $G_{t'}$ on the vertex set 
$V(G_{t'}) -W$.

In \PRe, the PDGs are used only in the analysis, and are not maintained by the algorithm. Here, in \FFDe, we will maintain a set of local data
structures for each PDG that is relevant to the current graph; also, in a small change of notation, we will denote a level graph for
time $t'\leq t$ as $\Gamma_{k'}$, where $k' = level(t')$ rather than the \PR notation of $\Gamma_{t'}$. These graphs are similar to the level graphs in \Thoe. 
As in \Thoe, 
  only certain information for $\levelgraph_k$ is explicitly maintained in its local data structures: the STs centered in level $k$ plus all the extensions that can generate STs in $\levelgraph_k$. The data structures used by our algorithm to maintain triples are Global and Local, which we now describe.

\begin{table}[ht]
	\begin{center}
		\begin{tabular}{|c|c|c|c|}
			\hline 
			\textbf{Notation} & \textbf{Data Structure} &  \multicolumn{2}{|c|} {\textbf{Appears}} \\ 
			\hline
			\hline 
			\multicolumn{4}{|c|}{\textbf{Part A :: Global Data Structures} {\scriptsize (for each pair of nodes $(x, y)$)}} \\ 
			\hline 
			\hline 
			$P(x,y)$ & all (centered) \LHT s from $x$ to $y$ with weight as key &  \multicolumn{2}{|c|}{ \multirow{2}{*}{\parbox[c]{3cm}{(\cite{DI04} for paths, \cite{NPR14b} for LSTs)}}} \\ 
			\cline{1-2} 
			$P^{*}(x,y)$ & all (centered) \HT s from $x$ to $y$ with weight as key & \multicolumn{2}{|c|}{ } \\ 
			\hline 
			$L(x,by)$ & $\{x' : (x'x, by)$ denotes a (centered) \LHTe $\}$ & \multicolumn{2}{|c|}{ \cite{PR14}, (\cite{NPR14b} for LSTs) } \\ 
			\hline 
			$R(xa, y)$ & $\{y' : (xa, yy')$ denotes a (centered) \LHTe $\}$ & \multicolumn{2}{|c|}{ \cite{PR14}, (\cite{NPR14b} for LSTs) }  \\ 
			\hline 
			\MT & global dictionary for Marking scheme &  \multicolumn{2}{|c|}{ \cite{PR14,NPR14b}}\\ 
			\hline 
			\hline
			\multicolumn{4}{|c|}{\textbf{Part B :: Local Data Structures} {\scriptsize (for each active level $i$, for each pair of nodes $(x,y)$)}} \\  
			\hline 
			\hline
			$P^{*}_i(x,y)$ & STs from $x$ to $y$ centered in level $i$ & \multicolumn{2}{|c|}{  \multirow{5}{*}{\parbox[c]{2.5cm}{(sketched in \cite{Thorup04} for paths)}}}\\ 
			\cline{1-2} 
			$L_i^{*}(x,y)$ & $\{x' : (x'x, y)$ denotes an $\ell$-tuple for SPs centered in level $i \}$ & \multicolumn{2}{|c|}{ } \\ 
			\cline{1-2} 
			$R_i^{*}(x,y)$ &  $\{y' : (x, yy')$ denotes an $r$-tuple for SPs centered in level $i \}$ & \multicolumn{2}{|c|}{ } \\ 
			\cline{1-2} 
			$LC_i^{*}(x,y)$ & the subset $\{ x' \in L_i^{*}(x,y) : x' \textrm{ is centered in level } i \}$  & \multicolumn{2}{|c|}{ }\\ 
			\cline{1-2} 
			$RC_i^{*}(x,y)$ &  the subset $\{ y' \in R_i^{*}(x,y) : y' \textrm{ is centered in level } i \}$ & \multicolumn{2}{|c|}{ }\\ 
			\hline
			$dict_i$ & dictionary of pointers from local STs to global $P$ and $P^*$ &  \multicolumn{2}{|c|}{new} \\
			\hline 
			\hline
			\multicolumn{4}{|c|}{\textbf{Part C :: Inherited Data Structures}} \\  
			\hline 
			\hline
			$\beta(\gamma)$ & flag bit for the (centered) triple $\gamma$ & \multicolumn{2}{|c|}{ \cite{PR14}}\\ 
			\hline
			$level(t)$ & level activated during $t$-th update &  \multicolumn{2}{|c|}{ \cite{PR14}}\\ 
			\hline 
			$time(k)$ & most recent update in which level $k$ is activated & \multicolumn{2}{|c|}{ \cite{PR14}}\\ 
			\hline 
			$\mathcal{N}$ & nodes (centered in levels) deactivated in the current step & \multicolumn{2}{|c|}{ \cite{PR14}}\\ 
			\hline 
			$\Gamma_k$ & level graph (PDG) created during $time(k)$-th update &  \multicolumn{2}{|c|}{ (\cite{PR14} as PDG, \cite{Thorup04} for paths)}\\ 
			\hline
			\hline
			\multicolumn{4}{|c|}{\textbf{Part D :: New Data Structures}} \\  
			\hline 
			\hline
			$C_\gamma$ & distribution of all paths in triple $\gamma$ among active levels & \multicolumn{2}{|c|}{new} \\ 
			\hline
			$\DMs(x,y)$ & linked-list containing the history of distances from $x$ to $y$ & \multicolumn{2}{|c|}{new}\\
			\hline
			$\LN(x,y,wt)$ & the set $\{b : \exists \, (xa,by)$ of weight $wt$ in $P(x,y)\}$ & \multicolumn{2}{|c|}{new}\\
			\hline
			$\RN(x,y,wt)$ & the set $\{a : \exists \, (xa,by)$ of weight $wt$ in $P(x,y)\}$ & \multicolumn{2}{|c|}{new}\\
			\hline  
		\end{tabular} 
	\end{center}
	\caption{Notation summary}
	\label{table:dsA}
\end{table}

\paragraph{\bf Global Structures.}
The global data structures are $P^*$, $P$, $L$ and $R$ (see Table \ref{table:dsA}, Part~A).
\begin{itemize}
	\item The structures $P^*(x,y)$ and $P(x,y)$ will contain HTs (including all STs) and LHTs (including all LSTs), respectively,  from $x$ to $y$. They are priority queues with the weights of the triple and a flag bit $\beta$ as key. For a triple $\gamma$ in $P$, the flag bit $\beta(\gamma)=0$ if the triple $\gamma$ is in $P$ but not in $P^*$, and $\beta(\gamma)=1$ if the triple $\gamma$ is in $P$ and $P^*$.
	\item The structure $L(x,by)$ ($R(xa,y)$) is the set of all left (right) extension vertices that generate a centered LHT in the \systemend.
\end{itemize}

\paragraph{\bf Local Structures.}
The local data structures we introduce in this paper are $L^*_i, R^*_i, LC^*_i$ and $RC^*_i$ (see Table \ref{table:dsA}, Part~B). These are generalization of the data structures sketched in \Tho for unique SPs in the graph. For every pair of nodes $(x,y)$:

\begin{itemize}
	\item The structure $P^*_i(x,y)$ contains the set of STs from $x$ to $y$ centered in $\levelgraph_i$. It is implemented as a set. 
	\item The structure $L^*_i(x,y)$ ($R^*_i(x,y)$) contains all left (right) extensions that generate a shortest $\ell$-tuple ($r$-tuple) centered in level $i$. It is implemented as a balanced search tree.
	\item The structure $LC^*_i(x,y)$ ($RC^*_i(x,y)$) contains left (right) extensions centered in level $i$ that generate a shortest $\ell$-tuple ($r$-tuple) centered in level $i$. It is implemented as a balanced search tree.
	\item  A dictionary $\dict_i$, contains STs in $P^*_i$ using the key $[x, y, a, b]$ and two pointers stored along with each ST. The two pointers refer to the location in $P(x,y)$ and $P^*(x,y)$
	of the triple of the form $(x,a,b,y)$ contained in $P^*_i(x,y)$.
\end{itemize}

\noindent
In order to recompute BC scores (see Section \ref{sec:fdbc}) we will consider $R^*(x,y)=\bigcup_i{R_i^*(x,y)}$ and similarly $L^*(x,y)=\bigcup_i{L_i^*(x,y)}$. 

\subsection{New Structures} \label{sec:newfeat}
We introduce two completely new data structures which are essential to achieve efficiency for our \FFD algorithm. Both are needed to address the Partial Extension Problem (PEP) which does not appear in \Tho (see section \ref{sec:ffeatures}).

\paragraph{\bf Distance History Matrix.}
 The \emph{distance history matrix} is a matrix $\DMs$ where each entry is a pointer to a linked list: for each $x, y \in V$, the linked list $\DMs(x,y)$ contains the sequence of different pairs $(wt,k)$, where each one represents an SP weight $wt$ from $x$ to $y$, along with the most recent level $k$ in which the weight $wt$ was the shortest distance from $x$ to $y$ in the graph $\levelgraph_k$. 
 The pair with weight $wt$ in $\DMs(x,y)$ is double-linked to every triple from $x$ to $y$ with weight $wt$ in the system. Precisely, when a new triple $\gamma$ from $x$ to $y$ of weight $wt$ is inserted in the algorithm, a link is formed between $\gamma$ and the pair $(wt,k)$ in $\DMs(x,y)$. With this structure, the \FFD algorithm can quickly check if there are still triples of a specific weight in the \systemend, especially for example when we need to remove a given weight from $\DMs$ (e.g. Step \ref{ffcleanup:remDM} - Alg. \ref{algo:ffcleanup-process}).
 Note that the size of each linked list is $O(\log n)$.

\paragraph{\bf Historical Extension (HE) Sets RN and LN.}
Another important type of structure we introduce are the sets $RN$ and $LN$. These structures are crucial to select efficiently the set of restored historical tuples that need to be extended (see Section \ref{ffixup-desc}). $RN(x,y,wt)$ ($LN(x,y,wt)$ works symmetrically) contains all nodes $b$ such that there exists at least one tuple of the form $(x \times,by)$ and weight $wt$ in $P(x,y)$. Similarly to $\DMs$, every time a new triple $\gamma$ of this form is inserted in the \systemend, a double link is created between $\gamma$ and the occurrence of $b$ in $RN(x,y,wt)$ in order to quickly access the triple when needed (e.g. Steps \ref{ffixup:RNstart} to \ref{ffixup:RNend} - Alg. \ref{algo:fget-n-paths}).

\vone
\noindent
The total space used by $\DMs$, $\RN$ and $\LN$ is $O(n^2 \log n)$. This is dominated by the overall space used by the algorithm to maintain all the triples in the \system across all levels (see Lemma \ref{lemma:count2}, Section \ref{sec:proof}).
\section{The \ffullydynamic Algorithm} \label{sec:algo}

Algorithm \ffullydynamic is similar to the fully dynamic algorithm in \PR and its overall description is given in Algorithm \ref{algo:ffully-main}.
The main difference is the introduction of the notion of levels as described in Section \ref{subsec:fully-impl}, and their activation/deactivation as in \Thoe.
At the beginning of the $t$-th update (with $k=level(t)$), we first activate the new level $k$ and we perform \ffullyupdate (Alg. \ref{algo:ffupdate}) on the updated node $v$. As in \PR and shown in Table \ref{table:dsA} - Part C, the set $\NODES$ consists of all vertices centered at these lower deactivated levels. All vertices in $\NODES$  are re-centered at level $k$ during the $t$-th update (Alg. \ref{algo:ffully-main}, Step \ref{fully-main:new2}), 
and `dummy' update operations are performed on each of these vertices. Note that 
$\NODES$ contains the $2^{k}-1$ most recently updated vertices in reverse order of update time (from the most recent to the oldest).
Procedure \ffullyupdate is invoked with the parameter $k$ representing the newly activated level.
Finally, all levels $j < k $  are deactivated (Alg. \ref{algo:ffully-main}, Step \ref{fully-main:new1}).

\begin{algorithm}
	\scriptsize
	\begin{algorithmic}[1]
		\STATE activate the new level $k$
		\STATE \ffullyupdateend($v,\weight',k$)
		\STATE generate the set $\NODES$ 
		\FOR {each $u \in \NODES$ in decreasing order of update time}
		\STATE \ffullyupdateend$(u,\weight',k)$ \COMMENT{dummy updates} \label{fully-main:new2}
		\ENDFOR
		\STATE deactivate all levels lower than $k$ \label{fully-main:new1}
	\end{algorithmic}
	\caption{\ffullydynamicend($G, v, \weight', k$)}
	\label{algo:ffully-main}
\end{algorithm}

\noindent
\paragraph{\bf \ffullyupdateend.}
As in \cite{DI04,NPR14b,PR14}, the update of a node occurs in a sequence of two steps: a \emph{cleanup phase} and a \emph{fixup phase}. Here, we call this update \ffullyupdate and it is a sequence of 2 calls: \ffullycleanup and \ffullyfixupend. Briefly, \ffullycleanup removes all \LHPe s in the \system containing the updated node (see Section \ref{ffcleanup-desc}), while \ffullyfixup identifies and adds the STs and LSTs in the updated graph that are not yet in the \system (see Section \ref{ffixup-desc}). Both \ffullycleanup and \ffullyfixup are more involved algorithms than their counterparts in \PR and the resulting algorithm will save a $O(\log n)$ factor over the amortized cost in \PRe.

\begin{algorithm}
	\scriptsize
	\begin{algorithmic}[1]
		\STATE \ffullycleanupend$(v)$
		\STATE \ffullyfixupend$(v,\weight',k)$
	\end{algorithmic}
	\caption{\ffullyupdateend$(v,\weight',k)$}
	\label{algo:ffupdate}
\end{algorithm}

\noindent
{\bf Dummy Updates.} 
Calling \ffullyupdate only on the updated nodes gives us a correct fully dynamic APASP algorithm. However, the number of \LHTe s generated could be very large making this strategy not efficient in general. Thus,
as in \PRe, the \ffullydynamic algorithm performs a sequence of dummy updates  
as follows. Consider an update on  $v$ at time $t$. Let $k= level(t)$. As shown in algorithm \ref{algo:ffully-main}, $\NODES{}$ is the set of $2^k -1$ most recently updated nodes $v_{t-1}, v_{t-2} , \cdots , v_j$, where $j= 2^k-1$. These are the nodes centered at levels smaller (more recent) than $k$ before the $t$-th update is applied. 
For each vertex $u$ in $\NODES$, starting from the most recent to the oldest, \ffullydynamic calls \ffullyupdate on the node $u$. This procedure removes all \LHTe s containing $u$ in every active level using \ffullycleanupend, and immediately reinserts them in the newly activated level $k$ by performing \ffullyfixup and using the local data structures. These are the
dummy updates. Dummy updates have the effect of removing from the \system any path $\pi$ containing a vertex in $\NODES$ that is no longer an LSP in the current graph, because $\pi$ is removed by the corresponding \ffullycleanup step and will not be restored during \ffullyfixupend.

We will establish in Section \ref{sec:proof} that \ffullydynamic
correctly updates the data structures with the amortized bound given in Theorem \ref{th:main}.

\subsection{Description of \ffullycleanup} \label{ffcleanup-desc}
\ffullycleanup removes all the LHPs going through the updated vertex $v$ from all the global structures $P$, $P^{*}$, $L$ and $R$, and from all local structures in any active level graph $\levelgraph_j$ that contains these triples. This involves decrementing the count of some triples or removing them completely (when all the paths in the triple go through $v$). The algorithm also updates local dictionaries and the $\DMs$, $\RN$ and $\LN$ structures.
Algorithm \ffullycleanup is a natural extension of the \NPR cleanup. An extension of the \NPR cleanup is used also in \PR but in a different way.

\begin{algorithm}
	\scriptsize
	\begin{algorithmic}[1]
		\STATE $H_c \leftarrow \emptyset$; \MT$\leftarrow \emptyset$ 
		\STATE $\gamma \leftarrow [(v,v), 0, 1]$; $\CA_{\gamma}[center(v)]=1$; add $[\gamma, \CA_{\gamma}]$ to $H_c$ \label{ffcleanup:init}
		\WHILE {$H_c \neq \emptyset$ } \label{ffcleanup:while}
		\STATE extract in $S$ all the triples with the same min-key $[wt,x,y]$ from $H_c$ \label{ffcleanup:extract}
		\STATE \ffullycleanupend-$\ell$-extend($S$,$[wt,x,y]$) (see Algorithm \ref{algo:ffcleanup-process})
		\STATE \ffullycleanupend-$r$-extend($S$,$[wt,x,y]$)
		\ENDWHILE
	\end{algorithmic}
	\caption{\ffullycleanupend$(v)$}
	\label{algo:ffcleanup}
\end{algorithm}

\ffullycleanup starts as in the \NPR algorithm. We add the updated node $v$ to $H_c$ (Step \ref{ffcleanup:init} -- Alg. \ref{algo:ffcleanup}) and we start extracting all the triples with same min-key (Step \ref{ffcleanup:extract} -- Alg. \ref{algo:ffcleanup}).
The main differences from \NPR start after we call Algorithm \ref{algo:ffcleanup-process}.
As in \cite{NPR14b}, we start by forming a new triple $\gamma'$ to be deleted (Steps \ref{ffcleanup:triplep} -- Alg. \ref{algo:ffcleanup}).
A new feature in Algorithm \ref{algo:ffcleanup-process} is to accumulate the paths that we need to remove level by level using the array $\CA'$. This is inspired by \Tho where unique SPs are maintained in each level. However, our algorithm maintains multiples paths spread across different levels using the $\CA_{\gamma}$ arrays associated to LSTs, and the technique used to update the $\CA_{\gamma}$ arrays is significantly different and more involved than the one described in \Thoe. Step~\ref{ffcleanup:updvect} - Alg.~\ref{algo:ffcleanup-process} calls \ffullyccenters (Alg.~\ref{algo:ffcleanup-vector}) that will perform this task.
\begin{algorithm}
	\scriptsize
	\begin{algorithmic}[1]
		\FOR {every $b$ such that $(x\times,by) \in S$}
		\STATE let $S_b \subseteq S$ be the set of all triples of the form $(x\times,by)$
		\STATE let $fcount'$ be the sum of all the $counts$ of all triples in $S_b$
		\FOR {every $x'$ in $L(x,by)$ s.t. $(x'x,by) \notin$ \MT}
		\STATE $wt' \leftarrow wt+\weight(x',x)$; $\gamma' \leftarrow ((x'x,by),wt', fcount')$ \label{ffcleanup:triplep}
		\STATE $\CA_{\gamma'} \leftarrow$ \ffullyccentersend$(\gamma', S_b)$ \label{ffcleanup:updvect}
		\STATE add $[\gamma', \CA_{\gamma'}]$ to $H_c$ \label{ffcleanup:addHcp}
		\STATE remove $\gamma'$ in $P(x',y)$ // decrements $count$ by $fcount'$ \label{ffcleanup:remP}
		\STATE set new center for $\gamma''=((x'x,by),wt')$ in $P(x',y)$ as $\arg\!\min_i(\CA_{\gamma''}[i] \neq 0)$
		\label{ffcleanup:setcen}	
		\IF { a triple for $(x'x,by)$ exists in $P(x',y)$} \label{ffcleanup:LRstart}
		
		\STATE insert $(x'x,by)$ in \MT
		\ELSE
		\STATE delete $x'$ from $L(x,by)$ and delete $y$ from $R(x'x,b)$ \label{ffcleanup:remLR}
		\ENDIF \label{ffcleanup:LRend}
		\IF { no triple for $((x'-,by),wt')$ exists in $P(x',y)$} \label{ffcleanup:RN}
		\STATE remove $b$ from $\RN(x',y,wt')$
		\ENDIF
		\IF { no triple for $((x'x,-y),wt')$ exists in $P(x',y)$} 
		\STATE remove $x$ from $\LN(x',y,wt')$	
		\ENDIF	\label{ffcleanup:LN}				
		\IF {a triple for $((x'x,by),wt')$ exists in $P^*(x',y)$} 
		\STATE remove $\gamma'$ in $P^*(x',y)$ // decrements $count$ by $fcount'$ \label{ffcleanup:remPS}
		\IF {$\gamma' \notin P^*(x',y)$}	
		\STATE remove the element with weight $wt'$ from $\DMs(x',y)$ if not linked to other tuples in $P^*(x',y)$ \label{ffcleanup:remDM}
		\ENDIF
		\FOR {each $i$} \label{ffcleanup:forPri}
		\STATE decrement $\CA_{\gamma'}[i]$ paths from $\gamma' \in P^*_i(x',y)$ \label{ffcleanup:updPri0}
		\IF {$\gamma'$ is removed from $P^*_i(x',y)$} \label{ffcleanup:remLSs}
		\IF	{$x'$ is centered in level $i$}
		\IF {$\forall \, j\geq i, P^*_j(x,y)=\emptyset$}
		\STATE remove $x'$ from $L^*_i(x,y)$ and remove $x'$ from $LC^*_i(x,y)$ \label{ffcleanup:updLC}
		\ENDIF
		\ELSIF {$P^*_i(x,y)=\emptyset$}
		\STATE remove $x'$ from $L^*_i(x,y)$
		\ENDIF
		\IF	{$y$ is centered in level $i$}
		\IF {$\forall \, j\geq i, P^*_j(x',b)=\emptyset$}
		\STATE remove $y$ from $R^*_i(x',b)$ and remove $y$ from $RC^*_i(x',b)$ \label{ffcleanup:updRC}
		\ENDIF
		\ELSIF {$P^*_i(x',b)=\emptyset$}
		\STATE remove $y$ from $R^*_i(x',b)$
		\ENDIF
		\ENDIF \label{ffcleanup:remLSe}
		\ENDFOR
		\ENDIF
		\ENDFOR	
		\ENDFOR
	\end{algorithmic}
	\caption{\ffullycleanupend-$\ell$-extend($S,[wt,x,y]$)}
	\label{algo:ffcleanup-process}
\end{algorithm}

\begin{algorithm}
	\scriptsize
	\begin{algorithmic}[1]
		\STATE let $\gamma' = ((x'x,by),wt',fcount')$ (the triple of the form $((x'x,by),wt')$ that contains all the paths through $v$ to be removed)	
		\STATE let $\gamma'' = ((x'x,by),wt',fcount'')$ (the triple of the form $((x'x,by),wt')$ in $P(x',y)$. Note that $\gamma'$ represents a subset of $\gamma''$) 	
		\STATE $j \leftarrow \arg\!\max_j(\CA_{\gamma''}[j]\neq 0)$ \label{ffcleanup:gpivot} // This is the oldest level in which a path in $\gamma''$ appeared for the first time
		\STATE $\CA' \leftarrow \sum_{\gamma \in S_b}{\CA_{\gamma}[r-1, \ldots, 0]}$ // This is the sum (level by level) of the triples of the form $(xa_i,by)$ that go through~$v$ and are extending to $\gamma'$ during this stage\label{ffcleanup:dist1}
		\STATE create a new center vector $\CA_{\gamma'}$ for the triple $\gamma'$ as follows 	\label{ffcleanup:newArr0}	
		\STATE \hspace{0.3in} for all the levels $m>j$ we set $\CA_{\gamma'}[m] = 0$ \label{ffcleanup:newArr1}
		\STATE \hspace{0.3in} for the level $j$ we set $\CA_{\gamma'}[j] = \sum_{k = j}^{r-1}{\CA'[k]}$ \label{ffcleanup:newArr2}
		\STATE \hspace{0.3in} for all the levels $i<j$ we set $\CA_{\gamma'}[i] = \CA'[i]$ \label{ffcleanup:newArr3}	
		\STATE $\CA_{\gamma''}[r-1, \ldots, 0] \leftarrow \CA_{\gamma''}[r-1, \ldots, 0] - \CA_{\gamma'}[r-1, \ldots, 0]$ // We update the $\CA$ vector for $\gamma'' \in P(x',y)$ \label{ffcleanup:Pvecupd}
		\STATE return $\CA_{\gamma'}$ // We return the correct vector for the generated $\gamma'$ triples
	\end{algorithmic}
	\caption{\ffullyccentersend$(\gamma', S_b)$}
	\label{algo:ffcleanup-vector}
\end{algorithm}

\ffullyccenters takes as input the generated triple $\gamma'$ of the form $(x'x,by)$ which contains all the paths going through the updated node $v$ to be removed, and the set of triples $S_b$ of the form $(x\times,by)$ that are extended to $x'$ to generate $\gamma'$. This procedure has two tasks: (1) generating the $\CA_{\gamma'}$ vector for the triple $\gamma'$ that will be reinserted in $H_c$ for further extensions, and (2) updating the $\CA_{\gamma''}$ vector for the tuple $\gamma''$ in $P(x',y)$ (note that $\gamma''$ is the corresponding triple in $P$ of $\gamma'$, before we subtract all the paths represented by $\gamma'$ level by level).\\
(1) - This task, which is more complex than the second task (which is a single step in the algorithm, see point (2) below), is accomplished in steps \ref{ffcleanup:dist1} to \ref{ffcleanup:newArr3}, Alg. \ref{algo:ffcleanup-vector} and uses the following technique.
In step \ref{ffcleanup:dist1} -- Alg. \ref{algo:ffcleanup-vector}, we store into the $\log n$-size array $C'$ the distribution over the active levels for the set of triples in $S_b$ that generates $\gamma'$ using the left extension to $x'$. In order to generate the correct vector $\CA_{\gamma'}$ (to associate with the triple $\gamma'$), we need to reshape the distribution in $\CA'$ according to the corresponding distribution of the triple $\gamma'' \in P$.  
The reshaping procedure works as follows: we first identify the oldest level $j$ in which the triple $\gamma''$ appeared in $P$ for the first time (Step \ref{ffcleanup:gpivot} -- Alg. \ref{algo:ffcleanup-vector}).
Recall that we want to remove $\gamma'$ paths containing $v$ from $\gamma''$, and $\gamma''$ does not exist in any level older than $j$. Vector $\CA'$ is the sum of $\CA_{\gamma}$ for all $\gamma \in S_b$ (Step \ref{ffcleanup:dist1} -- Alg. \ref{algo:ffcleanup-vector}). Those triples are of the form $(xa_i,by)$ and they could exist in levels older, equal or more recent than $j$. But the triples in $S_b$ that were present in a level older than $j$, were extended to $\gamma''$ in $P$ for the first time in level $j$. For this reason, step \ref{ffcleanup:newArr2} - Alg. \ref{algo:ffcleanup-vector} aggregates all the counts in $\CA'$ in levels older or equal $j$ in $\CA_{\gamma'}[j]$.
Moreover, for each level $i < j$, if a triple $\gamma \in S_b$ is present in the level graph $\Gamma_i$ with $count$ paths centered in level $i$, then $\levelgraph_i$ also contains its extension to $x'$ that is a subtriple of $\gamma''$ located in level $i$ with at least $count$ paths. Thus for each level $i < j$ step \ref{ffcleanup:newArr3} - Alg. \ref{algo:ffcleanup-vector}, copies the number of paths level-wise.   
This procedure allows us to precisely remove the LHPs only from the level graphs where they exist. After $\CA'$ is reshaped into $\CA_{\gamma'}$ (steps \ref{ffcleanup:newArr0} to \ref{ffcleanup:newArr3} - Alg. \ref{algo:ffcleanup-vector}), the algorithm returns this correct array for $\gamma'$ to Alg. \ref{algo:ffcleanup-process}.\\
(2) - This task is performed by the simple step \ref{ffcleanup:Pvecupd}, Alg. \ref{algo:ffcleanup-vector}, which is a subtraction level by level of LHPs.

After adding the new triple $\gamma'$ to $H_c$ (Step \ref{ffcleanup:addHcp} - Alg. \ref{algo:ffcleanup-process}), the algorithm continues as the \NPR (Steps \ref{ffcleanup:addHcp} to \ref{ffcleanup:remLR} -- Alg. \ref{algo:ffcleanup-process}) with some differences: we need to update centers, local data structures, $\DMs$, $\RN$ and $\LN$. We update the center of $\gamma'$ using $C_{\gamma'}$ (Step \ref{ffcleanup:setcen} - Alg. \ref{algo:ffcleanup-process}).
If $\gamma'$ is a shortest triple, we decrement the count of $\gamma' \in P^*(x',y)$ (Step \ref{ffcleanup:remPS} - Alg. \ref{algo:ffcleanup-process}). 
If $\gamma'$ is completely removed from $P^*(x',y)$ and $\DMs(x',y, wt')$ is not linked to any other tuple, we remove the entry with weight $wt'$ from $\DMs(x',y)$ (Step \ref{ffcleanup:remDM} - Alg. \ref{algo:ffcleanup-process}).
Moreover, we subtract the correct number of paths from each level using the (previously built) array $\CA_{\gamma'}$ (Step \ref{ffcleanup:updPri0} - Alg. \ref{algo:ffcleanup-process}). Finally for each active level $i$, if $\gamma'$ is removed from $P^*_i(x',y)$, we take care of the sets $L_i^*$ and $R_i^*$ (Steps \ref{ffcleanup:remLSs} to \ref{ffcleanup:remLSe} - Alg. \ref{algo:ffcleanup-process}). In the process, we also update $LC_i^*$ and $RC_i^*$ in case the endpoints of $\gamma'$ are centered in level $i$.
If $\gamma'$ is completely removed from $P(x',y)$, using the double links to the node $b$ in $\RN(x',y,wt')$, we check if there are other triples that use $b$ in $P(x',y)$ (Step \ref{ffcleanup:RN} - Alg. \ref{algo:ffcleanup-process}): if not we remove $b$ from $\RN(x',y,wt')$. 
A similar step handles $\LN(x',y,wt')$.

\subsection{Description of \ffullyfixup} \label{ffixup-desc}
\ffullyfixup is an extension of the fixup in \PR rather than \NPRe. This is because of the presence of the control bit $\beta$ (defined in Section \ref{subsec:fully-impl}), and the need to process historical triples (that are not present in \NPRe).
Algorithm \ffullyfixup will efficiently maintain exactly the LSTs and STs for each level graph in the \systemend. This is in contrast to \PRe, which can maintain \LHTe s that are not LSTs in any level graph (PDG).
\ffullyfixup maintains a heap $H_f$ of candidate LHTs to be processed in min-weight order. The main phase (Alg. \ref{algo:fffixup}) is very similar to the fixup in \PRe. The differences are again related to levels, centers and the new data structures. 

We start describing Algorithm \ref{algo:fffixup}.
We initialize $H_f$ by inserting the edges incident on the updated vertex $v$ with their updated weights (Steps \ref{algo:finit-f-start} to \ref{algo:finit-f-end} -- Alg. \ref{algo:ftrivial}), as well as a candidate min-weight triple from $P$ for each pair of nodes $(x,y)$ (Step \ref{fixup:phase2-begin} -- Alg. \ref{algo:ftrivial}). Then we process $H_f$ by repeatedly extracting collections of triples of the same min-weight for a given pair of nodes, until $H_f$ is empty (Steps \ref{ffixup:phase3-begin} to \ref{ffixup:phase3-end} -- Alg. \ref{algo:fffixup}). We will establish that the first set of triples for each pair $(x,y)$ always represents the shortest path distance from $x$ to $y$ (see Lemma \ref{fdfixcorr}), and the triple extracted are added to the \system if not already there (see Alg. \ref{algo:fget-n-paths} and Lemma \ref{proof:ffflem}). As in \PRe for efficiency, among all the triples present in the \system for a pair of nodes, we select only the ones that need to be extended: this task is performed by Algorithm \ref{algo:fget-n-paths} (this step is explained later in the description).
After the triples in $S$ are left and right extended by Algorithm \ref{algo:fprocess-f}, we set the bit $\beta(\gamma') = 1$ for each triple $\gamma'$ that is identified as shortest in $S$, since $\gamma'$ is correctly updated both in $P^*(x,y)$ and $P(x,y)$ (Step \ref{ffixup:setbit1} -- Alg. \ref{algo:fffixup}). Finally, we update the $\DMs(x,y)$ structure by inserting (or updating if an element with weight $wt$ is already present) the element with weight $wt$ and the current level at the end of the list (Step \ref{ffixup:addDM} -- Alg. \ref{algo:fffixup}). This concludes the description of Algorithm \ref{algo:fffixup}.

We now describe Algorithm \ref{algo:fget-n-paths} which is responsible to select only the triples that have valid extensions that will generate LHTs in the current graph. 
In Algorithm \ref{algo:fget-n-paths}, we distinguish two cases. When the set of extracted triples from $x$ to $y$ contains at least one path not containing $v$ (Step \ref{ffixup:phase3-main-check} -- Alg. \ref{algo:fget-n-paths}), then we process all the triples from $P(x,y)$ of the same weight. Otherwise, if all the paths extracted go through $v$ (Step \ref{ffixup:phase3-nomain-check} -- Alg. \ref{algo:fget-n-paths}), we only use the triples extracted from $H_f$. 

\begin{algorithm}
	\scriptsize
	\begin{algorithmic}[1]
		\STATE $H_f \leftarrow \emptyset$; \MT$\leftarrow \emptyset$ \label{ffixup:init-empty}
		\STATE \ffullypopulate$(v, \weight',k)$ \label{ffixup:phase2}
		\WHILE {$H_f \neq \emptyset$} \label{ffixup:phase3-begin}
		\STATE extract in $S'$ all the triples with min-key $[wt,x,y]$ from $H_f$ \label{ffixup:phase3-extract1}
		\IF {$S'$ is the first extracted set from $H_f$ for $x,y$}  \label{ffixup:phase3-first-ext}
		\STATE $S \leftarrow$ \ffullygetnew($S',P(x,y)$)	
		\STATE \ffullyfixupend-$\ell$-extend($S$,$[wt,x,y]$) (see Algorithm \ref{algo:fprocess-f}) \label{ffixup:lext}
		\STATE \ffullyfixupend-$r$-extend($S$,$[wt,x,y]$) \label{ffixup:rext}
		\STATE for every $\gamma \in S$ set $\beta(\gamma)=1$ \label{ffixup:setbit1}
		\STATE add an element with weight $wt$ and level $k$ to $\DMs(x,y)$ or update the level in the existing one \label{ffixup:addDM}
		\ENDIF
		\ENDWHILE \label{ffixup:phase3-end}
	\end{algorithmic}
	\caption{\ffullyfixupend$(v, \weight', k)$}
	\label{algo:fffixup}
\end{algorithm}

\begin{algorithm}
	\scriptsize
	\begin{algorithmic}[1]
		\FOR {each $(u, v)$}
		\STATE $\weight(u,v) = \weight'(u,v)$ \label{algo:finit-f-start}
		\IF {$\weight(u,v) < \infty$ }
		\STATE $\gamma = ((uv,uv),\weight(u,v),1)$; $C_\gamma[k] \leftarrow 1$
		\STATE update-num($\gamma$) $\leftarrow$ curr-update-num; num-v-paths($\gamma$) $\leftarrow 1$
		\STATE add $[\gamma, C_{\gamma}]$ to $H_f$ and $P(u,v)$
		\STATE add $u$ to $L(-,vv)$ and $v$ to $R(uu,-)$
		\ENDIF \label{algo:finit-f-end}
		\ENDFOR
		\FOR {each $(v,u)$}
		\STATE symmetric processing as Steps~\ref{algo:finit-f-start}--\ref{algo:finit-f-end} above
		\ENDFOR
		\FOR {each $x, y \in V$} \label{fixup:phase2-begin}
		\STATE add a min-key triple $[\gamma, C_{\gamma}] \in P(x,y)$ to $H_f$
		\ENDFOR \label{fixup:phase2-end}
	\end{algorithmic}
	\caption{\ffullypopulate$(v, \weight',k)$}
	\label{algo:ftrivial}
\end{algorithm}

\begin{algorithm}
	\scriptsize
	\begin{algorithmic}[1]
		\STATE $S \leftarrow \emptyset$; let $i$ be the min-weight level associated with $\DMs(x,y)$
		\IF {$P^*(x,y)$ increased min-weight after cleanup} \label{ffixup:phase3-main-check}
		\FOR {each $\gamma' \in S$ with-key $[wt,0]$} \label{ffixup:phase3-addfromP-begin}
		\STATE let $\gamma' = ((xa', b'y), wt, count')$ and $j=\arg\!\min_j(\CA_{\gamma'}[j] \neq 0)$ \label{ffixup:newcenter}
		\IF {$\gamma'$ is not in $P^*(x,y)$}                	
		\STATE add $\gamma'$ in $P^{*}(x,y)$ and $S$; add $x$ to $L^{*}(a',y)$ and $y$ to $R^{*}(x,b')$
		\STATE add $b'$ to $\RN(x,y,wt)$; place a double link between $\gamma'$ and $\DMs(x,y,wt)$ 
		\ELSIF {$\gamma'$ is in $P(x,y)$ and $P^*(x,y)$ with different counts} \label{fffixup:new1} 
		\STATE replace the count of $\gamma'$ in $P^{*}(x,y)$ with $count'$ and add $\gamma'$ to $S$ \label{fffixup:new2} 
		\ENDIF
		\STATE add $\gamma'$ to $P^*_j(x,y)$ and $dict_j$ \label{ffixup:phase3-addj}
		\STATE add $x$ to $L^*_j(a',y)$ and $y$ to $R^*_j(x,b')$
		\STATE add $x$ to $LC^*_j(a',y)$ ($y$ to $RC^*_j(x,b')$) if $x$ ($y$) is a level $i$ center \label{ffixup:phase3-addPe}
		\STATE add $\gamma'$ in $S$
		\ENDFOR
		\FOR {each $b' \in \RN(x,y,wt)$} \label{ffixup:RNstart}
		\IF {$\exists \, h<i : L_h^*(x,b')\neq \emptyset$ }
		\STATE add any $\gamma'$ of the form $(x \times,b'y)$ and weight $wt$ in $P^*(x,y)$ with $\beta(\gamma')=1$ to $S$ \label{ffixup:RNend}
		\ENDIF
		\ENDFOR
		\ELSE
		\FOR {each $\gamma' \in S'$ containing a path through $v$} \label{ffixup:phase3-nomain-check}
		\STATE let $\gamma' = ((xa', b'y), wt, count')$ and $k$ the current level
		\STATE add $\gamma'$ with paths$(\gamma',v)$ to $P^*(x,y)$, and $[\gamma', \CA_{\gamma'}]$ to $S$ \label{ffixup:phase3-add2PStar1}
		\STATE add $\gamma'$ to $P^*_k(x,y)$ and $dict_k$, $x$ to $L^*_k(a',y)$ and $y$ to $R^*_k(x,b')$
		\label{ffixup:phase3-addk}
		\STATE add $x$ to $LC^*_k(a',y)$ ($y$ to $RC^*_k(x,b')$) if $x$ ($y$) is a level $k$ center
		\label{ffixup:phase3-add2LRCStar1}
		\ENDFOR \label{ffixup:phase3-addfromX-end}	
		\ENDIF
		\STATE return $S$
	\end{algorithmic}
	\caption{\ffullygetnew($S', P_{xy}$)}
	\label{algo:fget-n-paths}
\end{algorithm}

Both cases have a similar approach but here we focus on the former which is more involved than the latter. 
As soon as we identify a new triple $\gamma'$ we compute its center $j$ by using its associated array $\CA_{\gamma'}$ (Step \ref{ffixup:newcenter} -- Alg. \ref{algo:fget-n-paths}). This is straightforward if compared to \ffullycleanup where we first need to update the center arrays. We add this triple to $P^*(x,y)$ and to $S$, which contains the set of triples that need to be extended.
We also add $\gamma'$ to $P_j^*(x,y)$ (Steps \ref{ffixup:phase3-addj} and \ref{ffixup:phase3-addk} -- Alg. \ref{algo:fget-n-paths}). We update $dict_j$ to keep track of the locations of the triple in the global structures. 
A similar sequence of steps takes place when all the extracted paths go through $v$ (Steps \ref{ffixup:phase3-nomain-check} to \ref{ffixup:phase3-addfromX-end} -- Alg. \ref{algo:fget-n-paths}). The only difference is that the local data structures to be updated are only the $\levelgraph_k$ data stuctures (Steps \ref{ffixup:phase3-addk} and \ref{ffixup:phase3-add2LRCStar1} -- Alg. \ref{algo:fget-n-paths}). 

A crucial difference from \PR and this algorithm is the way we collect the set $S$ of triples to be extended. Here we require the new HE data structures $\RN$ and $\LN$ (see Section \ref{sec:newfeat}) because of PEP instances (see Section \ref{sec:ffeatures}).
Let $i$ be the min-weight level associated with $DL(x,y)$.
For each node $b \in \RN(x,y,wt)$ we check if $L_h^*(x,b)$ contains at least one extension, for every $h < i$ (Steps \ref{ffixup:RNstart} to \ref{ffixup:RNend} -- Alg. \ref{algo:fget-n-paths}). In fact we need to discover all tuples with $\beta=1$ that are inside a PEP instance. In this instance, the triples restored as STs may or may not be extended. We cannot afford to look at all of them, thus our solution should check only the triples with an available extension. Moreover, all the extendable triples with with $\beta=1$ have extension only in levels younger than the level where they last appear as STs. 
Thus, we check for extensions only in the levels $h < i$. 

Using the HE sets, is the key to avoid an otherwise long search of all the valid extensions for the set of examined triples with $\beta=1$. In particular, without the HE sets, the algorithm could waste time by searching for extensions that are not even in the \systemend.
Correctness of this method is proven in section \ref{sec:proof}.
After the algorithm collects the set $S$ of triples that can be extended, \ffullyfixup calls \ffullyfixupend-$\ell$-extend (Alg. \ref{algo:fprocess-f}).

\begin{algorithm}
	\scriptsize
	\begin{algorithmic}[1]
		\FOR {every $b$ such that $(x\times,by) \in S$}
		\STATE let $S_b \subseteq S$ be the set of all triples of the form $(x\times,by)$
		\STATE let $fcount'$ be the sum of all the $counts$ of all triples in $S_b$; let $h$ be the $center(S_b)$
		\IF {$\exists \, \gamma \in S_b : \beta(\gamma) = 0$ }
		\STATE let $j$ be the level associated to the minweight $wt'>wt$ in $\DMs(x,y)$ \label{fprocess-f:DM}
		\FOR {every active level $h \leq i<j$}
		\FOR {every $x'$ in $L_i^{*}(x,b)$} \label{fprocess-f:exts}
		\IF {$(x'x,by) \notin$ \MT} \label{fprocess-f:start}
		\STATE $wt' \leftarrow wt+\weight(x',x)$; $\gamma' \leftarrow ((x'x,by),wt', fcount')$ \label{fprocess-f:Lexts}
		\STATE $\CA_{\gamma'} \leftarrow$ \ffullyfcentersend($S_b$); add $\gamma'$ to $H_f$ \label{fprocess-centers}
		\IF {a triple $\gamma''$ for $((x'x,by),wt')$ exists in $P(x',y)$}
		\STATE update the count of $\gamma''$ in $P(x',y)$ and $\CA_{\gamma''} = \CA_{\gamma''} + \CA_{\gamma'}$
		\STATE add $(x'x,by)$ to \MT
		\ELSE
		\STATE add $[\gamma',\CA_{\gamma'}]$ to $P(x',y)$; add $x'$ to $L(x,by)$ and $y$ to $R(x'x,b)$ \label{fprocess-f:addL}
		\ENDIF
		\STATE set $\beta(\gamma')=0$; set update-num($\gamma'$) \label{fprocess-f:end}
		\ENDIF
		\ENDFOR
		\ENDFOR \label{fprocess-f:exte}
		\FOR {every level $i < h$} \label{fprocess-f:extihs}
		\FOR {every $x'$ in $LC_i^{*}(x,b)$}
		\STATE execute steps \ref{fprocess-f:start} to \ref{fprocess-f:end}
		\ENDFOR
		\ENDFOR \label{fprocess-f:extihe}
		\ELSE
		\STATE let $j$ be the level associated to the minweight $wt$ in $\DMs(x,y)$ \label{fprocess-f:DM1} \label{fprocess-f:extis}
		\FOR {every level $i < j$} 
		\FOR {every $x'$ in $LC_i^{*}(x,b)$}
		\STATE execute steps \ref{fprocess-f:start} to \ref{fprocess-f:end}
		\ENDFOR
		\ENDFOR \label{fprocess-f:extie}
		\ENDIF
		\ENDFOR
	\end{algorithmic}
	\caption{\ffullyfixupend-$\ell$-extend($S$,$[wt,x,y]$)}
	\label{algo:fprocess-f}
\end{algorithm}

\begin{algorithm}
	\scriptsize
	\begin{algorithmic}[1]
		\STATE let $\CA'=\sum_{\gamma \in S_b}{\CA_{\gamma}}$ be the sum (level by level) of the new paths that are found shortest
		\STATE let $j$ be $\arg\!\max_j(\CA'[j]\neq 0)$, and $k=center(x')$	\label{fcenter:level}
		\IF {$k < j$}
		\STATE for all the levels $i<k$ we set $\CA_{\gamma'}[i] = \CA'[i]$ \label{fcenter:ress}
		\STATE for the level $k$ we set $\CA_{\gamma'}[k] = \sum_{q = k}^{r-1}{\CA'[q]}$
		\STATE for all the levels $m>k$ we set $\CA_{\gamma'}[m] = 0$ \label{fcenter:rese}
		\ELSE
		\STATE $\CA_{\gamma'} = \CA'$ \label{fcenter:nores}
		\ENDIF		
		\STATE return $\CA_{\gamma'}$
	\end{algorithmic}
	\caption{\ffullyfcentersend$(S_b)$}
	\label{algo:fffixup-centers}
\end{algorithm}

Here we describe the details of algorithm \ref{algo:fprocess-f}.
Its goal is to generate LHTs for the current graph $G$ by extending HTs.
Let $h$ be the \emph{center of $S_b$} defined as the most recent center among all the triples in $S_b$, and let $j$ be the level associated to the first weight $wt'$ larger than $wt$ in $\DMs(x,y)$.
The extension phase for triples is different from \PRe: in fact, the set of triples $S_b$ could contain only triples with $\beta(\gamma)=1$. In \PRe, the corresponding set $S_b$ contains only triples with $\beta(\gamma)=0$. We address two cases:\\ 
(\textbf{a}) -- If $S_b$ contains at least one triple $\gamma$ with $\beta(\gamma)=0$, we extend $S_b$ using the sets $L_i^{*}$ and $R_i^{*}$ with $h \leq i < j$ (Steps  \ref{fprocess-f:exts} to \ref{fprocess-f:exte} -- Alg.\ref{algo:fprocess-f}). In fact, the set $S_b$ contains at least one new path that was not extended in the previous iterations when $wt$ was the shortest distance from $x$ to $y$ (because of the $\beta(\gamma)=0$ triple). The LST generated in this way remains centered in level $h$. Moreover we extend $S_b$ also using the sets $LC_i^{*}$ and $RC_i^{*}$ with $i < h$ (Steps \ref{fprocess-f:extihs} to \ref{fprocess-f:extihe} -- Alg.\ref{algo:fprocess-f}). This ensures that every LST generated in a level $i$ lower than $h$ is centered in $i$ thanks to the extension node itself.  
This technique guarantees that each LHT generated by Algorithm \ref{algo:fprocess-f} is an LST centered in a unique level.\\
(\textbf{b}) -- In the case when there is no triple $\gamma$ in $S_b$ with $\beta(\gamma)=0$, then there is at least one extension to perform for $S_b$ and it must be in some level younger than the level where $wt$ stopped to be the shortest distance from $x$ to $y$ (this follows from the use of the HE sets in Alg. \ref{algo:fget-n-paths}). To perform these extensions we set $j$ as the level associated with the min-weight element in $\DMs(x,y)$, and we extend $S_b$ using the sets $LC_i^{*}$ and $RC_i^{*}$ with $i < j$ (Steps \ref{fprocess-f:extis} to \ref{fprocess-f:extie} -- Alg.\ref{algo:fprocess-f}). Again, every LHT generated is an LST centered in a unique level.
Finally, following \PRe, every generated LHT is added to $P$ and $H_f$ and we update global $L$ and $R$ structures. 
\begin{observation} \label{centerlht}
Every LHT generated by algorithm \ffullyfixup is an LST centered in a unique level graph.
\end{observation}
\begin{proof}
As described in (a) and (b) above, every LHT is generated using two triples which are shortest in the same level graph $\levelgraph_i$. Moreover, since at least one of them must be centered in level $i$, the resulting LHT is an LST centered in level $i$.
\hfill$\qed$	
\end{proof}	

The last novelty in the algorithm is updating center arrays (Alg. \ref{algo:fffixup-centers} called at step \ref{fprocess-centers} -- Alg. \ref{algo:fprocess-f}) in a similar way of \ffullycleanupend: Algorithm \ref{algo:fffixup-centers} identifies the oldest level $j$ related to the triples contained in $S_b$ (Step \ref{fcenter:level} -- Alg. \ref{algo:fffixup-centers}). 
If $j > k$ then we reshape the distribution for $\gamma'$ similarly to \ffullycleanup (Steps \ref{fcenter:ress} to \ref{fcenter:rese} -- Alg. \ref{algo:fffixup-centers}). Otherwise $\gamma'$ is completely contained in level $k$ and no reshaping is required (Step \ref{fcenter:nores} -- Alg. \ref{algo:fffixup-centers}).
\section{New Features in Algorithm \FFDe} \label{sec:ffeatures}
In this section, we discuss two challenges that arise when we attempt to generalize the level graph method used in \Tho (for APSP with unique SPs) to a fully dynamic APASP algorithm. Both are addressed by the algorithms in Section \ref{sec:algo} as noted below.

\noindent
\paragraph{\bf The bit $\beta$ feature.}
The control-bit $\beta$ was introduced (and only briefly described) in \PR to avoid the processing of untouched historical triples.
Here, we elaborate on this technique in more detail than \PRe, and we also describe how it helps in the more complex setting of the level \systemend.

Consider Figure \ref{fig:b}.
The ST $\gamma=((xa,by),wt,count)$ is created in level $k$ (Fig. \ref{fig:bk}). At $time(k)$, we have $\gamma \in P^*$ and also $\gamma \in P$ with $\beta(\gamma)=1$. In a more recent level $j < k$, a shorter triple $\gamma'=((xv,vy),wt', count')$, with weight $wt' < wt$, that goes trough an updated vertex $v$ is generated (Fig. \ref{fig:bj}). Thus at $time(j)$, we have $\gamma' \in P^*$ and also $\gamma' \in P$ with $\beta(\gamma')=1$; but $\gamma$ still appears in both $P^*$ and $P$ as a historical triple. Finally, a new LST $\gamma''=((xa',b'y),wt,count'')$, with the same weight as $\gamma$, is generated in level $i<j$ (Fig. \ref{fig:bi}). Note that $\gamma''$ is only in $P$ with $\beta(\gamma'')=0$ and not in $P^*$, as is the case of every LST that is not an ST. When an increase-only update removes $v$ and the triple $\gamma'$, the algorithm needs to restore all the triples with shortest weight $wt$. But while $\gamma$ is historical and does not require any additional extension, $\gamma''$ is only present in $P$ and needs to be processed. Our \FFD algorithm achieve this by checking the bit $\beta$ associated to each of these triples. The algorithm will extract and process all the triples with $\beta=0$ from $P(x,y)$ (see Step \ref{ffixup:phase3-addfromP-begin}, Alg.	\ref{algo:fget-n-paths}, where triples are extracted using $\beta=0$ as part of the key). These guarantees that a triple only present in $P$, or present in $P$ and $P^*$ with different counts is never missed by the algorithm.
\begin{figure}
	\centering
	\subfigure[level $k$]{
		\makebox[.3\textwidth]{
			\begin{tikzpicture}[every node/.style={circle, draw, inner sep=0pt, minimum width=5pt}]
			\node (x)[label=above:$x$] at (0,0)  {};
			\node (a)[label=left:$a$] at (0,-0.8) {};
			\node (b)[label=below left:$b$] at (0,-2.2) {};
			\node (y)[label=below:$y$] at (0,-3) {}; 
			\draw[->] (x) -- (a);
			\path[->,decoration={snake}] { (a) edge[decorate] (b)};
			\path[->,decoration={snake}] { (a) edge[bend right=60,decorate] (b.west)};
			\path[->,decoration={snake}] { (a) edge[bend left=60,decorate] (b.east)};
			\draw[->] (b) -- (y);
			\end{tikzpicture}} \label{fig:bk}} 
	\subfigure[level $j<k$]{
		\makebox[.3\textwidth]{
			\begin{tikzpicture}[every node/.style={circle, draw, inner sep=0pt, minimum width=5pt}]
			\node (x)[label=above:$x$] at (0,0)  {};
			\node (v)[label=right:$v$] at (1,-1.5) {};
			\node (y)[label=below:$y$] at (0,-3) {}; 
			\draw[->] (x) -- (v);
			\draw[->] (v) -- (y);
			\end{tikzpicture}} \label{fig:bj}} 
	\subfigure[level $i<j$]{
		\makebox[.3\textwidth]{
			\begin{tikzpicture}[every node/.style={circle, draw, inner sep=0pt, minimum width=5pt}]
			\node (x)[label=above:$x$] at (0,0)  {};
			\node (a1)[label=above left:$a'$] at (-1,-0.8) {};
			\node (b1)[label=below left:$b'$] at (-1,-2.2) {};
			\node (y)[label=below:$y$] at (0,-3) {}; 
			\draw[->] (x) -- (a1);
			\path[->,decoration={snake}] { (a1) edge[bend right=60,decorate] (b1.west)};
			\path[->,decoration={snake}] { (a1) edge[bend left=60,decorate] (b1.east)};
			\draw[->] (b1) -- (y);
			\end{tikzpicture}} \label{fig:bi}}
	\caption{The bit $\beta$ feature}
	\label{fig:b}
\end{figure}
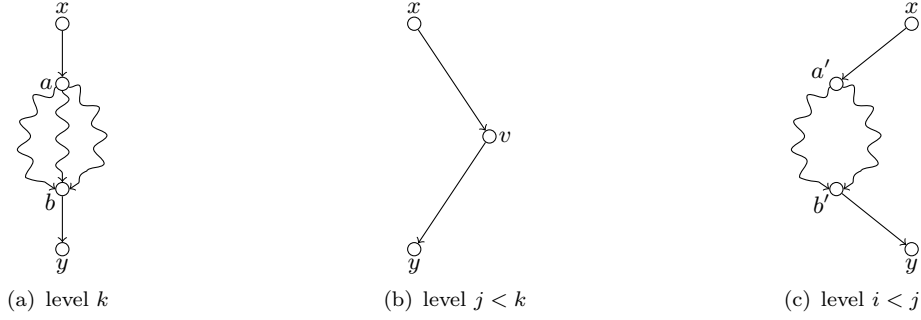 

\noindent
\paragraph{\bf The partial extension problem (PEP).}
Consider the update sequence described below and illustrated Fig. \ref{fig:ll}. 
Here the STs $\gamma=((xa,by),wt,count)$ and $\hat{\gamma}=((xa,cy),wt,count')$ 
are created in level $k$ (Fig. \ref{fig:llk}). Later, a left-extension to $x'$ generates the STs $\gamma'=((x'x,by),wt',count)$ and $\hat{\gamma}'=((x'x,cy),wt',count')$ in level $j<k$ (Fig. \ref{fig:llj}). Note that $\gamma$, $\hat{\gamma}$, $\gamma'$ and $\hat{\gamma}'$ are all present in $P^*$ and $P$ at $time(j)$. In a more recent level $i<j$, a decrease-only update on $v$ generates a shorter triple $\gamma_s=((xv,vy),wt_s, count_s)$ from $x$ to $y$, with $wt_s < wt$ going through $v$. 
In the same level, the triple $\gamma_s$ is also extended to $x'$ generating a triple $\gamma'' = ((x'x,vy),wt'',count_s)$ shorter than $\gamma'$ and $\hat{\gamma}'$ (Fig. \ref{fig:lli}). 
Thus at $time(i)$, $\gamma$, $\hat{\gamma}$, $\gamma'$ and $\hat{\gamma}'$ remain in $P^*$ as historical triples. 
Then, in level $h<i$, an update on $x''$ inserts the edges $(x'',x)$ and $(x'',c)$. This update generates an ST $\gamma'''=((x''c,cy))$ (shorter than $(x''x,vy)$) and also inserts $x'' \in LC_h^*(x,b)$) since $(x'',x)$ is on a shortest path from $x''$ to $b$; but it should not generate the triple of the form $(x''x,by)$ because $b$ is not on a shortest path from $x$ to $y$ at $time(h)$ (Fig. \ref{fig:llh}). 

\begin{figure}
	\centering
	\subfigure[level $k$]{
		\makebox[.20\textwidth]{
			\begin{tikzpicture}[every node/.style={circle, draw, inner sep=0pt, minimum width=5pt}]
			\node (x)[label=above:$x$] at (0,0)  {};
			\node (a)[label=left:$a$] at (0,-0.8) {};
			\node (b)[label=below left:$b$] at (0,-2.2) {};
			\node (c)[label=left:$c$] at (-1,-2.2) {};
			\node (y)[label=below:$y$] at (0,-3) {}; 
			\draw[->] (x) -- (a);
			\path[->,decoration={snake}] { (a) edge[decorate] (b)};
			\path[->,decoration={snake}] { (a) edge[decorate] (c)};
			\path[->,decoration={snake}] { (a) edge[bend left=60,decorate] (b.east)};
			\draw[->] (b) -- (y);
			\draw[->] (c) -- (y);
			\end{tikzpicture}} \label{fig:llk}} 
	\subfigure[level $j < k$]{
		\makebox[.20\textwidth]{
			\begin{tikzpicture}[every node/.style={circle, draw, inner sep=0pt, minimum width=5pt}]
			\node (x1)[label=above:$x'$] at (-0.5,0.8)  {};
			\node (x)[label=above right:$x$] at (0,0)  {};
			\node (a)[label=left:$a$] at (0,-0.8) {};
			\node (b)[label=below left:$b$] at (0,-2.2) {};
			\node (c)[label=left:$c$] at (-1,-2.2) {};
			\node (y)[label=below:$y$] at (0,-3) {}; 
			\draw[->] (x1) -- (x);
			\draw[->] (x) -- (a);
			\path[->,decoration={snake}] { (a) edge[decorate] (b)};
			\path[->,decoration={snake}] { (a) edge[decorate] (c)};
			\path[->,decoration={snake}] { (a) edge[bend left=60,decorate] (b.east)};
			\draw[->] (b) -- (y);
			\draw[->] (c) -- (y);
			\end{tikzpicture}} \label{fig:llj}} 
	\subfigure[level $i < j$]{
		\makebox[.20\textwidth]{
			\begin{tikzpicture}[every node/.style={circle, draw, inner sep=0pt, minimum width=5pt}]
			\node (x1)[label=above:$x'$] at (-0.5,0.8)  {};
			\node (x)[label=above right:$x$] at (0,0)  {};
			\node (v)[label=right:$v$] at (1,-1.5) {};
			\node (y)[label=below:$y$] at (0,-3) {}; 
			\draw[->] (x1) -- (x);
			\draw[->] (x) -- (v);
			\draw[->] (v) -- (y);
			\end{tikzpicture}} \label{fig:lli}}
	\subfigure[level $h < i$]{
		\makebox[.20\textwidth]{
			\begin{tikzpicture}[every node/.style={circle, draw, inner sep=0pt, minimum width=5pt}]
			\node (x2)[label=above:$x''$] at (0.5,0.8)  {};
			\node (x)[label=above left:$x$] at (0,0)  {};
			\node (a)[label=left:$a$] at (0,-0.8) {};
			\node (v)[label=right:$v$] at (1,-1.5) {};
			\node (b)[label=below left:$b$] at (0,-2.2) {};
			\node (c)[label=left:$c$] at (-1,-2.2) {};
			\node (y)[label=below:$y$] at (0,-3) {};
			\draw[->] (x2) -- (x);
			\path[->] { (x2) edge[bend right=45] (c)}; 
			\draw[->] (x) -- (a);
			\path[->,decoration={snake}] { (a) edge[decorate] (b)};
			\path[->,decoration={snake}] { (a) edge[bend left=60,decorate] (b.east)};
			\draw[->] (x) -- (v);
			\draw[->] (c) -- (y);
			\end{tikzpicture}} \label{fig:llh}} 
	\caption{PEP instance (only centered STs are kept in each level) -- all edge weights are unitary}
	\label{fig:ll}
\end{figure}
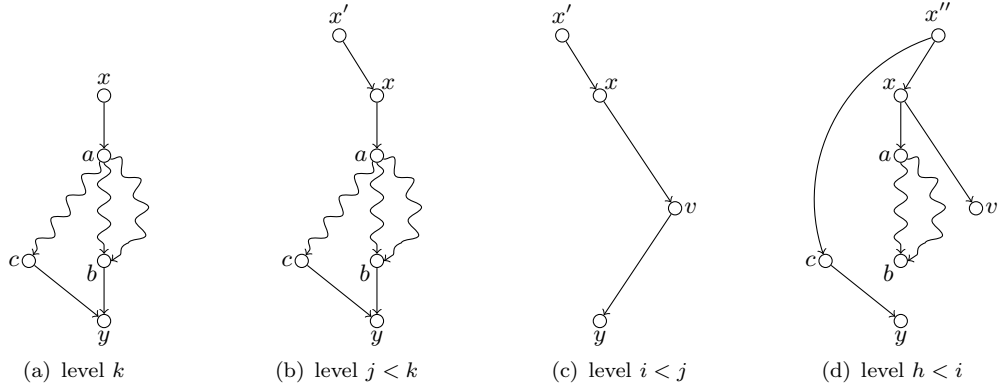 

When an increase-only update removes $v$ and the shortest triple ${\gamma}_s$ from $x$ to $y$, the algorithm needs to restore all  historical triples with shortest weights from $x$ to $y$. When $\gamma$ and $\hat{\gamma}$ are restored, we need to perform suitable left extensions as follows. An extension to $x''$ is needed only for $\gamma$: in fact $\hat{\gamma}$ should not be extended to $x''$ because the $\ell$-tuple of the newly generated tuple is not an ST in the graph.  On the other hand, no extension to $x'$ is needed since both $\gamma'$ and 
$\hat{\gamma}'$ will be restored (from HT to ST).  Our algorithm needs to distinguish all of these cases correctly and efficiently.

In order to maintain both correctness and efficiency in this scenario for APASP, we use two new data structures (described in Section \ref{sec:newfeat}): (1) the
historical distance matrix $\DMs$ that allow us to efficiently determine the most recent level graph in which an HT was an ST (see for example Steps \ref{fprocess-f:DM} and \ref{fprocess-f:DM1}, Alg. \ref{algo:fprocess-f}),
and (2) the HE sets $\LN$ and $\RN$ that allow us to efficiently identify exactly those new extensions that need to be performed (see for example Steps \ref{ffixup:RNstart} -- \ref{ffixup:RNend}, Alg. \ref{algo:fget-n-paths}).
The methodology of these data structures is fully discussed in the description of \ffullyfixup (Section \ref{ffixup-desc}).
Note that, the PEP doesn't arise in \Tho because of the unique SP assumption: in fact when only a single SP of a given length is present in the graph for each pair of nodes, the algorithm can check for all the $O(n^2)$ paths maintained in each level and decide which one should be extended. Given the presence of multiple SPs in our setting, we cannot afford to look at each tuple in the \system.
\section{Correctness and Complexity} \label{sec:proof}
In this section we will first prove the complexity bounds of our \ffullyupdate algorithm, then we will establish correctness.

\subsection{Complexity} \label{sec:complex}
The complexity analysis of algorithm \ffullydynamic is similar to that for the \PR algorithm. We highlight the following new elements:

\begin{enumerate}
	\item As noted in Section \ref{ffixup-desc} (see Observation \ref{centerlht}), every triple created by \ffullyfixup is an LST in the level graph (PDG) in which is centered, 
	and by the decremental only properties of level graphs, it will continue to be an LST in that level graph until it is removed. In contrast, \PR can create LHTs by combining HTs not centered in any PDG. This results in an additional $\Theta(\log n)$ factor in the amortized bound there. 
	
	\item  We can bound the number of LHTs that contain a given vertex $u$ as $O(z' \cdot \vstar^2)$, where $z'$ is the number of active level graphs that contain vertex $u$ and tuples passing through $u$ (by Corollary \ref{lemma:count1}). Given our level \systemend, $z'$ is clearly $O(\log n)$. In \PRe, this bound is $(z + z'^2)$ where $z$ is the number of active PDGs, and $z'$ is the number of PDGs that contain $v$.
		
	\item We can show that the number of accesses to $\RN$ and $\LN$, outside of the newly created tuples, is worst-case $O(n \cdot \nu^*)$ per call to \ffullyupdateend. The overhead given by the level data structures is $O(\log n)$ for each access (see Lemma \ref{lem:amor1}). These structures are not used in \PRe.
\end{enumerate}

\begin{lemma} \label{lemma:count}
	Let $G$ be a graph after a sequence of calls to \ffullyupdateend. Let $z$ be the number of active level graphs (PDGs), and let $z' \leq z$ be the number of level graphs that contain a given vertex $v$. Suppose that every HT in the \system
	is an ST in some level graphs, and every LHT is an LST in some level graph. If $n$ and $m$ bound the number of vertices and edges,
	respectively, in any of these graphs, and if
	$\nu^*$ bounds the maximum number of distinct edges that lie on shortest paths through any given vertex in any of the these graphs, then:
	
	\begin{enumerate}
		\item The number of LHTs  in $G$'s \system is at most $O(z \cdot m \cdot \nu^*)$.
		
		\item The number of LHTs that contain a vertex $v$ in $G$ is $O( z' \cdot {\nu^*}^2)$.
	\end{enumerate}
\end{lemma}

\begin{proof}
	For part 1, we bound the number of LHTs $(xa,by)$ (across all weights) that can exist in
	$G$. The edge $(x,a)$ can be chosen in $m$ ways, and once we fix $(x,a)$, the 
	$r$-tuple $(a,by)$ must be an ST in one of the $\levelgraph_j$. Since $(b,y)$ must lie on  
	a shortest path through
	$a$ centered in a graph $\levelgraph_i$, that contains the $r$-tuple $(a,by)$ of shortest weight in $\levelgraph_i$, the number of different choices
	for $(b,y)$ that will then uniquely determine the tuple $(xa, by)$, together with its weight, is
	$z \cdot \nu^*$. Hence the number of LHTs in $G$'s \system is $O(z \cdot m \cdot \nu^*)$.
	
	For part 2, the number of LHTs that contain $v$ as an internal vertex is simply the number of
	LSTs across the $z'$ graphs that contains $v$, and this is  $O(z' \cdot \vstar^2)$. We now
	bound the number of LHTs $(va, by)$. There are $n-1$ choices for the edge $(v,a)$ and
	$z' \cdot \nu^*$ choices for the $r$-tuple $(a, by)$, hence the total number of such tuples 
	is $O(z' \cdot n \cdot \nu^*)$. The same bound holds for LHTs of the form $(xa, bv)$. Since
	$\nu^* = \Omega (n)$, the result in part 2 follows.
	\hfill$\qed$
 \end{proof}
 
  \begin{corollary} \label{lemma:count1}
  	At a given time step, let $B$ be the maximum number of tuples in the \system containing a path through
  	a given vertex in a given level graph.
	Then, $B=O(\vstar^2)$.
  \end{corollary}
  
 \begin{lemma}
 	\label{lem:amor}
 	(a) - The cost for an \ffullycleanup call on a node $v$ when $z'$ active levels contain triples through $v$  is $O(z' \cdot \vstar^2 \cdot \log n)$.\\
 	(b) - The cost for a real \ffullycleanup call is $O(\vstar^2 \cdot \log^2 n)$\\
 	(c) - The cost for a dummy \ffullycleanup call is $O(\vstar^2 \cdot \log n)$.
 \end{lemma} 
 \begin{proof}
 	(a) - Algorithm \ref{algo:ffcleanup} extracts all the \LHTe s that go through the update vertex from the heap $H_c$. Since the number of these \LHTe s is bounded by $B$ at each level (by Corollary \ref{lemma:count1}), the total cost is $O(z' \cdot B \log n)$ where $z'$ is the number of active levels that contain triples through $v$.
 	Algorithm \ref{algo:ffcleanup-vector} requires only $O(\log n)$ time for each step, except in step \ref{ffcleanup:dist1} where the cost is $O(\log n)$ for each triple extracted from $H_c$ that goes through the updated vertex. Since the number of such triples is bounded by $O(z' \cdot \vstar^2)$ (by Lemma \ref{lemma:count}), the worst-case cost for a call to Algorithm \ref{algo:ffcleanup-vector} within an \ffullycleanup phase is $O(z' \cdot \vstar^2 \cdot \log n)$.
 	In Algorithm \ref{algo:ffcleanup-process}, a triple can be added to heap $H_c$, or searched and removed from a constant number of priority queues among $z'$ different active levels. Moreover, for the structures $\DMs$, $\RN$ and $\LN$ each triple spends a constant time to be unlinked and eventually to update the structures. Since, priority queue operations have a $O(\log n)$ cost and the number of triples examined is bounded by $O(z' \cdot \vstar^2)$, the complexity for Algorithm \ref{algo:ffcleanup-process} is at most $O(z' \cdot \vstar^2 \cdot \log n)$. Thus an \ffullycleanup call that operates on $z'$ active levels requires at most $O(z' \cdot \vstar^2 \cdot \log n)$.\\
 	(b) - Since $z \leq \log 2n$, the cost for a real \ffullycleanup call is $O(\vstar^2 \cdot \log^2 n)$ (by part (a)).\\
 	(c) - For a dummy cleanup on a vertex $w$, \ffullycleanup only needs to clean the local data structures in level $center(w)$, where $w$ is centered, and in the current level graph. In fact, let $t$ be the current update step; in the dummy cleanup phase, we start with the node $u$  that was updated at time $t-1$ (the most recent update before the current one). The node $u$ received an update in the previous phase, thus it disappeared from all the levels older than $level(t-1)$ and, with it, all the LSTs containing $u$ in these levels. Hence, all the triples containing $u$ in the \system must be LSTs in $level(t-1)$. We have at most $B$ of them and \ffullycleanup spends $O(B \cdot \log n)$ (considering the access to the data structures) to remove them.   
 	Then, the dummy update reinserts $u$ only in the current graph. The next phase moves on the node $u'$ updated at time $t-2$. Again, all the tuples containing $u'$ must be LSTs in $level(t-2)$ and eventually the current graph if they were inserted because of the previous dummy update on $u$.
 	
 	Suppose in fact that there is a tuple $\gamma$ that contains $u'$ in another level (except the current graph). The tuple $\gamma$ cannot be in a level older than $level(t-2)$ because when $u'$ was updated at time $(t-2)$, the cleanup algorithm removed all the tuples containing $u'$ from any level older than $t-2$. Moreover, a tuple containing $u'$ present in a level younger than $level(t-2)$ could appear if and only if it was generated by any update more recent of $t-2$ (in this case only the dummy update on $u$ performed in the current graph). Thus a contradiction.
 	
 	This argument can be recursively applied to  every other node in the sequence: in fact for the node $u''$ updated at time $(t-i)$ all the nodes updated in the interval $[t-i+1,t-1]$ will be already processed by \ffullycleanupend, leaving all the tuples containing $u''$ only in $level(t-i)$ and $t$. It follows that, for a dummy update, $z' = 2$.
 	Thus the cost for a dummy \ffullycleanup call is $O(\vstar^2 \cdot \log n)$ (by part (a)).
 	\hfill$\qed$
 \end{proof}
 
  \begin{lemma}
  	\label{lem:amor1d}
  	The cost for a dummy \ffullyfixup call on a node $v$ is $O(\vstar^2 \cdot \log n)$. 
  \end{lemma} 
  \begin{proof}
  	Consider a dummy \ffullyfixup applied to a vertex $v$ in $\NODES$. We only need
  	to bound the cost for accessing the entries in the $P^*(x,y)$ and the cost of re-adding LSTs containing $v$, previously removed by the dummy \ffullycleanup but still in the current graph after the dummy update.
  	In fact the vertex $v$ is removed by an earlier dummy \ffullycleanupend, and while this removes all the HPs containing the vertex $v$, it does not change any LST centered in any $\levelgraph_j$ that does not contain $v$. Hence these other
  	LSTs will be present in the \system with unchanged weight and count, when dummy \ffullyfixup is applied to $v$.
  	Since for any pair $x,y$, the SP distance will not change after the dummy update, the dummy \ffullyfixup will only insert in the set $S$ triples containing the node $v$ for additional extension (this is accomplished by the check at Step \ref{ffixup:phase3-main-check} - Alg. \ref{algo:fget-n-paths}, followed by Steps \ref{ffixup:phase3-nomain-check}--\ref{ffixup:phase3-addfromX-end}, Alg. \ref{algo:fget-n-paths}).
  	Hence, only the LSTs containing $v$ in the current $level(t)$ graph will be processed and added to the \systemend, and there are at most $B$ of them (by Corollary \ref{lemma:count1}). Thus a dummy \ffullyfixup for any $v$ needs to access $P^*$ for each pair of nodes, and reinsert at most $B$ tuples (containing $v$) in the current graph. Hence the overall complexity for a dummy \ffullyfixup is  $O((n^2 + B) \cdot \log n) = O(\vstar^2 \cdot \log n)$.
  	\hfill$\qed$
  \end{proof}

  We now address the complexity of a real \ffullyfixup call. We first define the concept of a \emph{triple pair} that will be used in lemma \ref{lem:amor1}  to establish the bound for a real \ffullyfixup call.
  Finally, we complete our analysis by presenting a proof of Theorem~\ref{th:main}. 
 
  \begin{definition}
  	If $C_{\gamma}[i] \geq 1$ then $(\gamma, i)$ is a \emph{triple pair} in the \systemend.
  	If $(\gamma,i)$ is not a triple pair in the \system at the start of step $t$ but is a triple pair after the update at time step $t$, then $(\gamma, i)$ is  a \emph{newly created triple pair}  at time step $t$.
  	\end{definition}
  
   \begin{lemma} \label{lemma:count2}
   	At a given time step,
   	let $D$ be
   	the number of triple pairs in the level \systemend.
   	Then,
   	\begin{enumerate}
   		\item The value of $D$ is at most $O(m \cdot \nu^* \cdot \log n)$.
   		\item  The space used is $O(m \cdot \nu^* \cdot \log n)$.
   	\end{enumerate}
   \end{lemma}
   
   \begin{proof}
   	
   	1. Every $C_{\gamma}[i] \geq 1$ represents a distinct LST in $\Gamma_i$, hence the result follows since the number of levels is $O(\log n)$ and the
   	number of LSTs in a graph is $O(\vstar \cdot m^*)$.
   	
   	\noindent
   	2. Since every triple is of size $O(1)$, the memory used by our \FFD algorithm is dominated by $D$, and result follows from 1.
   	\hfill$\qed$
   \end{proof}	 
   
 \begin{lemma}
 	\label{lem:amor1}
 	The cost for a real \ffullyfixup call is $O  (\vstar^2 \cdot \log^2 n + X  \cdot \log n)$ , where $X$ is the number of newly created triple pairs after the update step.
 \end{lemma} 
 \begin{proof}
  A triple is accessed only a constant number of time during \ffullyfixup for a cost of $O(\log n)$, so it suffices to establish that the
  number of existing triples accessed during the call is $O(\vstar^2 \cdot \log n)$.
  
   There are only $O(n^2)$ accesses to triples in the call to Algorithm 8 in line 2 of 	\ffullyfixup since 
 	$O(n^2)$ entries in the global $P^*(x,y)$ structures are accessed to initialize $H_f$. This takes $O(n^2 \cdot \log n)$ time after considering
	the $O(\log n)$ cost per data structure operation. We now address the accesses made in the main loop. We will distinguish two cases and they will be charged to $X$ as follows.
	
	\vhalf
	\noindent
	\underline{1: $\beta(\gamma)=0$} -- In the main loop of Algorithm 7, starting in Step 3, any triple $\gamma$ that is accessed with $\beta (\gamma) = 0$ is an
	 LST at some level $i$ where it is not identified as an ST in $\Gamma_i$. 
	 During this call, in Steps \ref{ffixup:phase3-addfromP-begin}--\ref{ffixup:phase3-addPe}, Alg. \ref{algo:fget-n-paths} (or Steps \ref{ffixup:phase3-nomain-check}--\ref{ffixup:phase3-addfromX-end}, Alg. \ref{algo:fget-n-paths} if the distance for the endpoints of $\gamma$ did not change)
	 $\gamma$ is added as an ST in level $i$, and will never be removed as an ST for level $i$ until it is removed from the tuple system
	 (due to the fact that $\Gamma_i$ is a purely decremental graph). Since $\gamma$ with $\beta (\gamma) =0$ is a newly added triple to level $i$, then the pair $(\gamma,i)$ is a newly created triple pair at step $t$. Hence, we can charge $(\gamma,i)$ to $X$ in this call of \ffullyfixupend.
	
	\vhalf
	\noindent
	\underline{2: $\beta(\gamma)=1$} -- We now consider  triples accessed that have $\beta = 1$. This is the most nontrivial part of our analysis since even though any
	 such triple $\gamma$ must exist with the same count in every level in both $P$ and $P^*$, we may still need to form some extensions
	 since the triple may have been an HT when extension vertices were updated, and hence these extension may not have been
	 performed. Here is where the $LN$ and $RN$ sets are accessed, and we now analyze the cost of these accesses. (The
	 correctness of the associated steps is analyzed in the next section.)
	 
	 Let $j$ be the most recent level in which $\gamma$ was an ST in $G$ and assume we are dealing with left extensions (right extensions are symmetrical). Now that $\gamma$ is restored, the only case (Steps \ref{ffixup:RNstart} to \ref{ffixup:RNend}, Alg. \ref{algo:fget-n-paths}) in which we need to process it is when there exists a left extension for the $\ell$-tuple of $\gamma$ to a node $x'$ centered in a level $i$ more recent than $j$. In fact, the LST generated by this extension will appear for the first time centered in level $i$, hence the pair $(\gamma,i)$ is a newly created triple pair at step $t$ and we can charge its creation to $X$.
	 We now show how our HE sets efficiently handle this case.
	By steps \ref{ffixup:RNstart} to \ref{ffixup:RNend}, Alg. \ref{algo:fget-n-paths}, our algorithm only processes a restored triple $\gamma$ with $\beta(\gamma)=1$ when it has at least one centered extension in some active level younger than the level in which $\gamma$ was shortest for the last time. We can bound the total computation for these steps as follows: for a given $x$, 
	$RN(x,y,wt)$ contains a node $b$ for every incoming edge to $y$ in one of the SSSP dags (historical and shortest) rooted at $x$. 
	Since we can extend in at most $O(\log n)$ active levels during any update and the size of a single dag is at most $\nu^*$, these steps take time $O(\nu^*\cdot n\log n)$ throughout the entire update computation. 
	\hfill$\qed$
 	 \end{proof}	 
 We can now establish the proof of our main theorem. 	 

\noindent
{\bf Proof of Theorem~\ref{th:main}.}
Consider a sequence $\Sigma$ of $r= \Omega (n)$ calls to algorithm \FFDe. Recall that the data structure is reconstructed after every $2n$ steps,
so we can assume $r=\Theta (n)$. These $r$ calls to \FFD make $r$ real calls to \ffullyupdateend, and also make additional dummy updates.
As in \PRe,
 	across the $r$ real updates in $\Sigma$,
 	the algorithm performs $O(r \log n)$ dummy updates.
 	This is because $r/2^k$ real updates are performed at level $k$ during the entire computation, and each such update is
 	accompanied by $2^k-1$ dummy updates. So, across all real updates there are $O(r)$
 	dummy updates per level, adding up to $O(r \log n)$ in total, across the $O(\log n)$ levels.
	
	When \ffullycleanup is called on a vertex $v$ for a dummy update, $z'=2$ since $v$ can be present 
	 only in the most recent current level and the level at which it is centered. (This is because every vertex that was centered
	 at a more recent level than $v$ has already been subjected to a dummy update, and hence all of these vertices are now
	 centered in the current level.)
 	Thus, by Lemma \ref{lem:amor}, each \ffullycleanup for a dummy update has cost $O(B \cdot \log n)$. By Lemma \ref{lem:amor1d}, a call to  \ffullyfixup for a dummy update
 	has cost $O(\vstar^2 \cdot \log n)$. Thus the total cost is $O((\vstar^2 \cdot \log n)\cdot r\log n)$ across all dummy updates.
	Also, the number of tuples accessed by all of the dummy update calls to \ffullycleanupend, and hence the number of tuples
	removed by all dummy updates, is $O(r \cdot \vstar^2 \cdot \log n)$.
	
	 For the real calls to \ffullyfixupend, let $X_i$ be the   number of newly added triple pairs in the $i$th real call to \ffullyfixupend. Then by
	Lemma~\ref{lem:amor1}, the cost of this $i$th call is $O(\vstar^2 \cdot \log^2 n + X_i \cdot  \log n)$. Let $X= \sum_{i=1}^r X_i$. Hence the total
	cost for the $r$ real calls to  \ffullyfixup is $O(r \cdot \vstar^2 \cdot \log^2 n + X \cdot \log n)$. We now bound $X$ as follows: $X$ is no more than
the maximum number of triples that can remain in the system after $\Sigma$ is executed, plus the number of tuples $Y$ removed
	from the \systemend. 	Tuples are removed only in calls to \ffullycleanupend. The total number removed by $r \log n$ dummy calls is 
	$O(r \cdot  \log n \cdot \vstar^2)$ (by Lemma \ref{lem:amor}). The total number removed by the $r$ real calls is $O(r \cdot \vstar^2 \cdot \log n)$ (by Lemma \ref{lem:amor}).
	Hence $Y= O(r \cdot \vstar^2 \cdot \log n)$. Clearly the maximum number of triples in the \system is no more than $D$, which
	counts the number of triple pairs, and we have 
	$D= O(m \cdot \vstar \cdot \log n) = O(n^2 \cdot \vstar \cdot \log n)$ (by Lemma \ref{lemma:count2}). Since  $r=\Theta (n)$, we have
	$D=O(r \cdot n \cdot \vstar \cdot \log n)$, and this is dominated by $Y$ since $\vstar = \Omega (n)$.
	Hence the cost of the $r$ calls to \FFD is $O(r \cdot \vstar^2 \cdot \log^2 n)$ (after factoring in the $O(\log n)$ cost per tuple access), and 
		hence the amortized cost of  each call to \ffullyupdate is $O(\vstar^2 \cdot \log^2 n)$.

\subsection{Correctness}
For the correctness, we assume that all the global and local data structures are correct before the update, and we will show the correctness of them after the update.

\paragraph{\bf Correctness of Cleanup -} The correctness of \ffullycleanup is established in Lemma \ref{lemma:fclean1}. We will prove that all paths containing the updated vertex $v$ are removed from the \systemend. Moreover, the center of each triple is restored, if necessary, to the level containing the most recently updated node on any path in this triple. 
Note that (as in \cite{DI04,NPR14b,PR14}) at the end of the cleanup phase, the global structures $P$ and $P^*$ may not have all the LHTs in $G \setminus \{v\}$.

\begin{lemma} \label{lemma:fclean1}
	At the end of the cleanup phase triggered by an update on a vertex $v$, every LHP that goes through $v$ is removed from the global structures. Moreover, in each level graph $\levelgraph_i$, each SP that goes through $v$ is removed from $P^*_i$. For each level $i$, the local structures $L_i^*$, $R_i^*$, $RC_i^*$ and $LC_i^*$ contain the correct extensions; the global structures $L$ and $R$ contain the correct extensions, for each $r$-tuple and $\ell$-tuple respectively, and the structures $\RN$ and $\LN$ contain only nodes associated with tuples in $P$. The $\DMs$ structure only contains historical distances represented by at least one path in the updated graph.
	Finally, every triple in $P$ and $P^*$ has the correct updated center for the graph $G \setminus \{v\}$.
\end{lemma}

\begin{proof}
	To prove the lemma statement, we use a loop invariant on the while loop in Step~\ref{ffcleanup:while} of Algorithm~\ref{algo:ffcleanup}. 
	We show that the while loop maintains the following invariants.
	
	{\noindent \bf Loop Invariant:} At the start of each iteration of the while loop in Step~\ref{ffcleanup:while} of Algorithm~\ref{algo:ffcleanup},
	assume that the first triple to be extracted from $H_c$ and processed has min-key = $[wt, x, y]$. Then the following properties hold about the \system and $H_c$.
	\begin{enumerate}
		\item \label{proof:item1} For any $a, b \in V$, if $G$ contains $c_{ab}$ LHPs of weight ${wt}$ of the form $(xa, by)$
		passing through $v$,
		then $H_c$ contains a triple $\gamma = ((xa, by), { wt}, c_{ab})$ with key $[wt,x,y]$ already processed: the $c_{ab}$ LHPs through $v$ are not present in the \systemend.
		
		\item \label{proof:claim2} Let  $[\hat{wt},\hat{x},\hat{y}]$ be the last key extracted from $H_c$ and processed before $[wt,x,y]$. For any key $[wt_1, x_1, y_1] \leq [\hat{wt},\hat{x},\hat{y}]$, let 
		$G$ contain ${ c}  > 0 $ number of LHPs of weight  ${ wt_1}$ of the
		form $(x_1 \times, b_1y_1)$. Further, let ${ c_v}$ (resp. ${ c_{\bar v}}$) denote the number of such LHPs
		that pass through $v$ (resp. do not pass through $v$).
		Here ${ c_v + c_{\bar v} = c}$. For every extension $x' \in L(x_1, b_1y_1)$, let $wt' = wt_1 + \weight (x',x_1)$ be the weight of the extended triple $(x'x_1, b_1y_1)$. Then, (the following assertions are similar for $y' \in R(x_1a_1, y_1)$)\\
		\textbf{Global Data Structures: }
		\begin{enumerate}
			\item \label{proof:item2} if $c>c_{v}$ there is a triple in $P(x', y_1)$ 
			of the form $(x'x_1, b_1y_1)$ and weight $wt'$ representing $c-c_v$ LHPs. Moreover, its center is updated according to the last update on any path represented by the triple. If $c=c_v$ there is no such triple in $P(x', y_1)$.
			\item \label{proof:item21}
			If a triple of the form $(x'x_1, b_1y_1)$ and weight $wt'$ is present as an HT in $P^*(x', y_1)$, then it represents the exact same number of LHPs $c-c_v$ of the corresponding triple in $P(x',y_1)$. This is exactly the number of HPs of the form $(x'x_1, b_1y_1)$ and weight $wt'$ in $G \setminus \{v\}$.
			\item \label{proof:item3} $x' \in L(x_1, b_1y_1)$,  $y_1 \in R(x'x_1, b_1)$, and    $(x'x_1, b_1y_1) \in $ \MT
			iff ${ {c_{\bar v}} > 0}$. 
			\item \label{proof:item4} A triple corresponding to $(x' x_1, b_1y_1)$
			with weight $wt'$ and counts $c_v$ is in $H_c$.
			A similar assertion holds for $y' \in R(x_1a_1, y_1)$.
			\item \label{proof:item6} The structure $\RN(x',y_1,wt')$ contains a node $b$ iff at least one path of the form $(x' \times,by_1)$ and weight $wt'$ is still represented by a triple in $P(x', y_1)$. A similar assertion holds for a node $a$ in $\LN(x',y_1,wt')$.
			\item \label{proof:item7} If there is no HT of the form $(x'x, b_1y_1)$ and weight $wt'$ in $P^*(x', y_1)$ then the entry $\DMs(x', y_1)$ with weight $wt'$ does not exists.
		\end{enumerate}
		\textbf{Local Data Structures:} for each level $j$, let $c_j$ be the number of LSPs of the form $(x'x_1, b_1y_1)$ and weight $wt'$ centered in $\levelgraph_j$ and let $c_j(v)$ be the ones that go through $v$. Thus $c=\sum_{j}{c_j}$ and $c_v=\sum_{j}{c_j(v)}$. Then,
		\begin{enumerate}[resume]
			\item \label{proof:litem1} the value of $\CA_\gamma[j]$, where $\gamma$ is the triple of the form $(x'x_1, b_1y_1)$ and weight $wt'$ in $P(x',y_1)$, is $c_j-c_j(v)$. 
			\item \label{proof:litem2}
			If a triple $\gamma$ of the form $(x'x_1, b_1y_1)$ and weight $wt'$ is present as an HT in $P^*$, then $P_j^*(x',y_1)$ represents only $c_j-c_j(v)$ paths. If $c_j-c_j(v)=0$ then the link to $\gamma$ is removed from $dict_j$. Moreover, $x' \in L_j^*(x_1, y_1)$ (respectively $LC_j^*(x_1, y_1)$ if $x'$ is centered in $\levelgraph_j$) iff $x'$ is part of a shortest path of the form $(x'x_1, \times y_1)$ centered in $\levelgraph_j$. A similar statement holds for $y_1 \in R_j^*(x', b_1)$ (respectively $RC_j^*(x', b_1)$ if $y_1$ is centered in $\levelgraph_j$).
		\end{enumerate}
		\item \label{proof:item5} For any key $[wt_2, x_2, y_2 ] \geq [wt, x, y]$, let
		$G$ contain $c > 0$ LHPs of weight  ${ wt_2}$ of the
		form $(x_2a_2, b_2y_2)$. 
		Further, let ${ c_v}$ (resp. ${ c_{\bar v}}$) denote the number of such LHPs
		that pass through $v$ (resp. do not pass through $v$).
		Here ${ c_v + c_{\bar v} = c}$. 
		Then the tuple $(x_2 a_2, b_2 y_2) \in$ \MTend, iff $c_{\bar v} > 0$ and a triple for
		$(x_2a_2, b_2 y_2)$ is present in $H_c$
	\end{enumerate}

		\noindent {\bf Initialization:} We start by showing that the invariants hold before the first loop iteration.
		The min-key triple in $H_c$ has key $[0, v, v]$. Invariant assertion~$\ref{proof:item1}$
		holds since  we inserted into
		$H_c$ the trivial triple of weight $0$ corresponding to the vertex $v$
		and that is the only triple of such key. Moreover, since we do not represent trivial
		paths containing the single vertex, no counts need to be decremented.
		Since we assume positive edge weights, there are no LHPs
		in $G$ of weight less than zero. Thus all the points of invariant assertion~$\ref{proof:claim2}$ hold trivially.
		Invariant assertion~$\ref{proof:item5}$ holds since $H_c$ does not contain any triple of weight $> 0$ and we initialized \MT to empty.
		
		\noindent {\bf Maintenance:} Assume that the invariants are true before an iteration $k$ of the loop.
		We prove that the invariant assertions remain true before the next iteration $k+1$.
		Let the min-key triple at the beginning of the $k$-th iteration be $[wt_k, x_k, y_k]$.
		By invariant assertion~$\ref{proof:item1}$, we know that for any $a_i, b_j$, if there exists a triple $\gamma$ of the form $(x_k a_i, b_j y_k)$
		of weight $wt_k$ representing $count$ paths containing $v$, then it is present in $H_c$.
		Now consider the set of triples with key $[wt_k, x_k, y_k]$
		which we extract in the set $S$ (Step~\ref{ffcleanup:extract}, Algorithm~\ref{algo:ffcleanup}).
		We consider left-extensions of triples in $S$; symmetric arguments apply for right-extensions.
		Consider for a particular~$b$ the set $S_{b} \subseteq S$ of triples of the form $(x_k-,by_k)$,
		and let $fcount'$ denote the sum of the counts of the paths represented by triples in $S_b$.
		Let $x' \in L(x_k, by_k)$ be a left extension; our goal is to generate the triple $\gamma'$ of the form $(x'x_k, b y_k)$ with count $fcount'$ and weight $wt' = wt_k+\weight(x',x_k)$, and an associated vector $\CA(\gamma')$ that specifies the distribution of paths represented by $\gamma'$ level by level. These paths will be then removed by the algorithm.
		However, we generate such triple only if it has not been generated by a right-extension of another set of paths by checking the \MT structure:	we observe that the paths
		of the form $(x'x_k, by_k)$ can be generated by right extending to $y_k$ the set of
		triples of the form $(x'x_k, \times b)$. Without loss of generality assume that the triples of
		the form $(x'x_k, \times b)$ have a key which is greater than the key $[wt_k, x_k, y_k]$. Thus,
		at the beginning of the $k$-th iteration, by invariant assertion~$\ref{proof:item5}$, we know that $(x'x_k, by_k) \notin$ \MTend.
		Step~\ref{ffcleanup:triplep}, Alg.~\ref{algo:ffcleanup} creates a triple $\gamma'$ of the form $(x'x_k, by_k)$ of weight $wt'$ and $fcount'$.
		
		The set of triples in $S_b$ can have different centers and we are going to remove (level by level) paths represented by $\gamma'$. To perform this task we consider the vector $\CA_{\gamma''}$: it contains the full distribution of the triple $\gamma'' \in P(x'y_k)$ of the form $((x'x_k,by_k),wt')$ and indicates the oldest level $j$ in which $\gamma''$ was generated for the first time. This level is exactly $\arg\!\max_j(\CA_{\gamma''}[j]\neq 0)$ and it is identified in Step \ref{ffcleanup:gpivot} - Alg. \ref{algo:ffcleanup-process}. All the paths represented in $S_b$ that are centered in some level $m$ older or equal to $j$ were extended for the first time in level $j$ to generate $\gamma''$. Moreover, each path centered in a level $i$ younger than $j$ was extended in level $i$ itself.
		Thus, we can compute a new center vector $\CA_{\gamma'}$ (according to the distribution in $\CA_{\gamma''}$) of the paths containing $v$ that we want to delete at each active level, as in steps \ref{ffcleanup:newArr0} to \ref{ffcleanup:newArr3} - Alg. \ref{algo:ffcleanup-vector}. In step \ref{ffcleanup:Pvecupd} - Alg. \ref{algo:ffcleanup-vector} the vector $\CA_{\gamma''}$ is updated: the paths are removed level by level according to the new distribution. This establishes invariant assertion $\ref{proof:litem1}$.
		
		The triple $\gamma'$ is immediately added to $H_c$ with $\CA_{\gamma'}$ for further extensions (Step. \ref{ffcleanup:addHcp} - Alg. \ref{algo:ffcleanup-process}).
		This establishes invariant assertions $\ref{proof:item4}$.
		Thus we reduce the counts of $\gamma'$ in $P(x',y_k)$ by $fcount$ (Step. \ref{ffcleanup:remP} - Alg.\ref{algo:ffcleanup-process}) and we set the new center for the remaining tuple $\gamma''$ in $P(x',y_k)$ establishing invariant assertion $\ref{proof:item2}$.
		Steps \ref{ffcleanup:LRstart} to \ref{ffcleanup:LRend} - Alg. \ref{algo:ffcleanup-process} check if there is any path of the form $(x-,by)$ that can use $x'$ as an extension. In this case we add $\gamma'$ to the \MTend. If not, we safely remove the left and right extension ($x'$ and $y$) from the \systemend. This establishes invariant assertion $\ref{proof:item3}$. If $\gamma'$ is an HT in $P^*(x',y_k)$, we decrement its count (Step. \ref{ffcleanup:remPS} - Alg.\ref{algo:ffcleanup-process}) establishing invariant assertion $\ref{proof:item21}$.
		In steps \ref{ffcleanup:RN} to \ref{ffcleanup:LN} - Alg. \ref{algo:ffcleanup-process}, we use the double links between $b \in \RN(x',y_k,wt')$ and tuples to efficiently check if there are other triples linked to $b$; 
		if not we remove $b$ from $\RN(x',y_k,wt')$ establishing invariant assertion $\ref{proof:item6}$.
		Using a similar double link method with the structure $\DMs(x',y_k)$, we establish invariant assertion $\ref{proof:item7}$ after step \ref{ffcleanup:remDM} - Alg. \ref{algo:ffcleanup-process}.
		
		To operate in the local data structures we require $\gamma'$ to be an HT in $P^*(x',y_k)$. Using the previously created vector $\CA_{\gamma'}$, we reduce the count associated with $\gamma' \in 1P^*_i(x',y_k)$ for each level $i$ (Step \ref{ffcleanup:updPri0} - Alg. \ref{algo:ffcleanup-process}). 
		After the above step, if there are no paths left in $P^*_i(x',y_1)$ then there are no STs of the form $(x'x_1,by_1)$ centered in level $i$.
		In this case we remove the extension $x'$ and $y_k$ from the local structures of level $i$. 
		This is done in steps \ref{ffcleanup:remLSs} to \ref{ffcleanup:remLSe} -Alg. \ref{algo:ffcleanup-process}: in case $x'$ is not centered in level $i$, then any path in $\gamma'$ centered in level $i$ is generated by a node centered in level $i$ located between $x_k$ and $y_k$. Thus if any SP from $x'$ to $y_k$ (that uses $(x',x_k)$ as a first edge) remains in in $\levelgraph_i$, it must be also counted in $P^*_i(x_k,y_k)$. Thus, we remove $x'$ from $L^*_i(x_k,y_k)$ only if $P^*_i(x_k,y_k)$ is empty. In the case $x'$ is centered in level $i$ and $P^*_i(x_k,y_k)$ is empty, $x'$ could still be the extension of other paths from $x_k$ to $y_k$ centered in levels older than $i$. The algorithm checks them all and if they do not exist in older levels we can safely remove $x'$ from $LC^*_i(x_k,y_k)$ (Step \ref{ffcleanup:updRC}, Alg. \ref{algo:ffcleanup-process}). A similar argument holds for the right extension $y_1$. 
		This establishes invariant assertion $\ref{proof:litem2}$ and completes claim \ref{proof:claim2}.
		
		When any triple is generated by a left extension (or symmetrically
		right extension), it is inserted into $H_c$ as well as into \MTend. This establishes invariant assertion~$\ref{proof:item5}$ at the beginning of the $(k+1)$-th iteration.
		
		Finally, to see that invariant assertion~$\ref{proof:item1}$ holds at the beginning of the $(k+1)$-th iteration, let the
		min-key at the $(k+1)$-th iteration be $[wt_{k+1}, x_{k+1}, y_{k+1}]$. Observe that triples
		with weight $wt_{k+1}$ starting with $x_{k+1}$ and ending in $y_{k+1}$ can be created
		either by left extending or right extending the triples of smaller weight. And since for each of
		iteration $\le k$, invariant assertion~$\ref{proof:claim2}$ holds for any extension, we conclude that invariant assertion~$\ref{proof:item1}$ holds at the beginning of the $(k+1)$-th iteration.
		This concludes our maintenance step.
		
		\noindent {\bf Termination:} The condition to exit the loop is $H_c = \emptyset$. Because invariant assertion $\ref{proof:item1}$ maintains in $H_c$ all the triples already processed, then $H_c = \emptyset$ implies that there are no other triples to extend in the graph $G$ that contain the updated node $v$. Moreover, because of invariant assertion $\ref{proof:item1}$, every triple containing the node $v$ inserted into $H_c$ has been correctly decremented from the \systemend. Remaining triples have the correct update center because of invariant $\ref{proof:item2}$. Finally, for invariant assertions $\ref{proof:litem1}$ and $\ref{proof:litem2}$, the structures $L_i^{*}, LC_i^{*}, R_i^{*}, RC_i^{*}$ are correctly maintained for every active level $i$ and the paths are surgically removed only from the levels in which they are centered. This completes the proof.
		\hfill$\qed$
\end{proof}

\paragraph{\bf Correctness of Fixup -} For the fixup phase, we need to show that the triples generated by our algorithm are sufficient to maintain all the ST and LST in the current graph $G$.
As in \PRe, we first show in the following lemma that \ffullyfixup computes all the correct distances for each pair of nodes in the updated graph. 
Finally, we show that data structures and counts are correctly maintained at the end of the algorithm (Lemma \ref{proof:ffflem}).
\begin{lemma} \label{fdfixcorr} 
	For every pair of nodes $(x,y)$, let $\gamma = ((xa, by), wt, count)$ be one of the min-weight triples from $x$ to $y$ extracted from $H_f$ during \ffullyfixupend. Then $wt$ is the shortest path distance from $x$ to $y$ in $G$ after the update.
\end{lemma}

\begin{proof}
	Suppose that the lemma is violated. Thus, there will be an extraction from $H_f$ during \ffullyfixup such that the set of extracted triples $S'$, of weight $\hat{wt}$ is not shortest in $G$ after the update. Consider the earliest of these events when $S'$ is extracted from $H_f$. Since $S'$ is not a set of STs from $x$ to $y$, there is at least one shorter tuple from $x$ and $y$ in the updated graph. 
	Let $\gamma' = ((xa',b'y),wt,count)$ be this triple that represents at least one shortest path from $x$ to $y$, with $wt < \hat{wt}$. 
	Since $S'$ is extracted from $H_f$ before any other triple from $x$ to $y$, $\gamma'$ cannot be in $H_f$ at any time during \ffullyfixupend. Hence, it is also not present in $P(x,y)$ as an LST at the beginning of the algorithm, otherwise it (or another triple with the same weight) would be placed in $H_f$ by step \ref{ffixup:phase2} - Alg. \ref{algo:fffixup}. Moreover, if $\gamma'$ is a single edge (trivial triple), then it was already an LST in $G$ present in $P(x,y)$ before the update, and it is added to $H_f$ by step \ref{fixup:phase2-begin} - Alg. \ref{algo:ftrivial}; moreover since all the edges incident to $v$ are added to $H_f$ during steps~\ref{algo:finit-f-start} to \ref{algo:finit-f-end} of Alg.~\ref{algo:ftrivial}, then $\gamma'$ must represent SPs of at least two edges. We define $left(\gamma')$ as the set of LSTs of the form $((xa',c_ib'), wt-\weight(b',y), count_{c_i})$ that represent all the LSPs in the left tuple $((xa',b'), wt-\weight(b',y))$; similarly we define $right(\gamma')$ as the set of LSTs of the form $((a'd_j,b'y), wt-\weight(x,a'), count_{d_j})$ that represent all the LSPs in the right tuple $((a',b'y), wt-\weight(x,a'))$. 
	
	Observe that since $\gamma'$ is an ST, all the LSTs in $left(\gamma')$ and $right(\gamma')$ are also STs. A triple in $left(\gamma')$ and a triple in $right(\gamma')$ cannot be present in $P^*$ together at the beginning of \ffullyfixupend. 
	In fact, if at least one triple from both sets is present in $P^*$ at the beginning of \ffullyfixupend, then the last one inserted during the fixup phase triggered during the previous update, would have generated an LST of the form $((xa',b'y), wt)$ automatically inserted, and thus present, in $P$ at the beginning of the current fixup phase (a contradiction). Thus either there is no triple represented by $left(\gamma')$ in $P^*$, or there is no triple represented by $right(\gamma')$ in $P^*$.
	
	Assume w.l.o.g. that the set of triples in $right(\gamma')$ is placed into $P^*$ after $left(\gamma')$ by \ffullyfixupend. 
	Since edge weights are positive, $wt-\weight(x,a') < wt < \hat{wt}$, and because all the extractions before $\gamma$ were correct, then the triples in $right(\gamma')$ were correctly extracted from $H_f$ and placed in $P^*$ before the wrong extraction of $S'$.
	Let $i$ be the level in which $left(\gamma')$ is centered, and let $j$ be the level in which $right(\gamma')$ is centered.
	By the assumptions, all the triples in $left(\gamma')$ are in $P^*$ and we need to distinguish 3 cases:
	\begin{enumerate}
		\item if $j=i$, then \ffullyfixup generates the tuple $((xa',b'y),wt)$ in the same level and place it in $P$ and $H_f$.
		\item if $i>j$, the algorithms \ffullyfixup extends the set $right(\gamma')$ to all nodes in $L_i^*(a',b')$ for every $i \geq j$ 
		(see Steps \ref{fprocess-f:exts} to \ref{fprocess-f:end} - Alg. \ref{algo:fprocess-f}). Thus, since  $left(\gamma')$ is centered in some level $i > j$, the node $x$ is a valid extension in $L_i^*(a',b')$, making the generated $\gamma'$ an LST in $\levelgraph_j$ that will be placed in $P(x,y)$ and also into $H_f$ (during Step \ref{fprocess-centers} - Alg. \ref{algo:fprocess-f}).
		\item if $j>i$, then $x$ was inserted in a level younger than $i$. In fact, all the paths from $a'$ to $b'$ must be the same in $right(\gamma')$ and $left(\gamma')$ otherwise the center of $right(\gamma')$ should be $i$. Hence, the only case when $j>i$ is when the last update on $left(\gamma')$ is on the node $x$ in a level $i$ younger than $j$. Thus $x \in LC_i^*(a',b')$. But \ffullyfixup extends $right(\gamma')$ to all nodes in $LC_i^*(a',b')$ for every $i < j$, placing the generated LST $\gamma'$ in $P(x,y)$ and also into $H_f$  
		(see Steps \ref{fprocess-f:extihs} to \ref{fprocess-f:extihe} - Alg. \ref{algo:fprocess-f}).
	\end{enumerate} 
	
	Thus the algorithm would generate the tuple $((xa',b'y),wt)$ (as a left extension) and place it in $P$ and $H_f$ (because all the triples in $left(\gamma')$ are already in $P^*$). Therefore, in all cases, a tuple $((xa',b'y),wt)$ should have been extracted from $H_f$ before any triple in $S'$. A contradiction.
	\hfill$\qed$
\end{proof}

\begin{lemma} \label{proof:ffflem}
	After the execution of \ffullyfixupend, for any $(x, y) \in V$, the sets $P^{*}(x,y)$ ($P(x,y)$) contains all the SPs (LSPs) from $x$ to $y$ in the updated graph.
	Also, the global structures $L, R$ and the local structures $P_i^*, L_i^{*},R_i^{*},LC_i^{*},RC_i^{*}$ and $dict_i$ for each level $i$ are correctly maintained. The structures $\RN$ and $\LN$ are updated according to the newly identified tuples. The $\DMs$ structure contains the updated distance for each pair of nodes in the current graph. Finally, the center of each new triple is updated.
\end{lemma}
\begin{proof}
	We prove the lemma statement by showing the following loop invariant. Let $G'$ be the graph after the update.
	
	{\noindent \bf Loop Invariant:} 
	At the start of each iteration of the while loop in Step~\ref{ffixup:phase3-begin} of \ffullyfixupend,
	assume that the first triple in $H_f$ to be extracted and processed has min-key = $[wt, x, y]$. Then the following properties hold about the \system and $H_f$.
	\begin{enumerate}
		\item \label{proof:fitem1} For any $a, b \in V$, if $G'$ contains $c_{ab}$ SPs of form $(xa, by)$ and weight ${wt}$,
		then $H_f$ contains a triple of form $(xa, by)$ and weight $wt$ to be extracted and processed. Further, a triple $\gamma = ((xa, by), { wt}, c_{ab})$ is present in $P(x,y)$. 
		\item Let $[\hat{wt},\hat{x},\hat{y}]$ be the last key extracted from $H_f$ and processed before $[wt,x,y]$. For any key $[wt_1, x_1, y_1] \leq [\hat{wt},\hat{x},\hat{y}]$, let
		$G'$ contain ${ c}  > 0 $ number of LHPs of weight  ${ wt_1}$ of the
		form $(x_1a_1, b_1y_1)$. Further, let ${ c_{new}}$ (resp. ${ c_{old }}$) denote the number of these LHPs
		that are {\em new} (resp. not {\em new}).
		Here ${ c_{new} + c_{old} = c}$. If $c_{new} > 0$ then,\\
		\textbf{Global Data Structures:}
		\begin{enumerate}
			\item \label{proof:fitem2}
			there is an LHT $\gamma$ in $P(x_1, y_1)$ of the form $(x_1a_1, b_1y_1)$ and weight $wt_1$ that represents $c$ LHPs, with an updated center defined by the last update on any of the paths represented by the LHT. 
			\item \label{proof:fitem21}
			If a triple of the form $(x_1a_1, b_1y_1)$ and weight $wt_1$ is present as an HT in $P^*$, then it represents the exact same count of $c$ HPs of its corresponding triple in $P$. This is exactly the number of HPs of the form $(x_1a_1, b_1y_1)$ and weight $wt_1$ in $G'$. Its control bit $\beta$ is set to 1.
			\item \label{proof:fitem3} $x_1 \in L(a_1, b_1y_1)$, $y_1 \in R(x_1a_1, b_1)$.
			Further, $(x_1a_1, b_1y_1) \in $ \MT
			iff ${ {c_{old}} > 0}$.
			\item \label{proof:fitem4} If $\beta(\gamma)=0$ or $\beta(\gamma)=1$ and there is an extension $x' \in L_j^*(x_1, y_1)$ that generates a centered LST in a level $j$, an LHT corresponding to $(x' x_1, b_1y_1)$
			with weight $wt'~=~wt_1~+~\weight (x',x_1) \geq wt$ and counts equal to the sum of new paths represented by its constituents, is in $H_f$ and $P$. 
			A similar assertion holds for an extension $y' \in R_j^*(x_1, y_1)$.
			\item \label{proof:fitem6} The structure $\RN(x_1,y_1,wt_1)$ contains a node $b$ iff at least one path of the form $(x_1 \times,by_1)$ and weight $wt_1$ is represented by a triple in $P(x_1, y_1)$. A similar assertion holds for a node $a$ in $\LN(x_1,y_1,wt_1)$.
			\item \label{proof:fitem7} The entry $\DMs(x_1, y_1)$ with weight $wt_1$ is updated to the current level.
		\end{enumerate}
		\textbf{Local Data Structures:} for each level $j$, let $c_j$ be the number of SPs of the form $(x_1a_1, b_1y_1)$ and weight $wt_1$ centered in $\levelgraph_j$ and let $c_j(n)$ be the new ones discovered bythe algorithm. Thus $c=\sum_{j}{c_j}$ and $c_{new}=\sum_{j}{c_j(n)}$. Then,
		\begin{enumerate}[resume]
			\item \label{proof:flitem2} the value of $\CA_\gamma[j]$, where $\gamma$ is the triple of the form $(x_1a_1, b_1y_1)$ and weight $wt_1$ in $P(x_1,y_1)$, is $c_j$. 
			\item \label{proof:flitem21}
			If a triple $\gamma$ of the form $(x_1a_1, b_1y_1)$ and weight $wt_1$ is present as an HT in $P^*$, then $P_j^*(x1,y1)$ represents $c_j$ paths. A link to $\gamma$ in $P$ is present in $dict_j$. Moreover, $x_1 \in L_j^*(a_1, y_1)$ (respectively $LC_j^*(a_1, y_1)$ if $x_1$ is centered in $\levelgraph_j$). A similar statement holds for $y_1 \in R_j^*(x_1, b_1)$ (respectively $RC_j^*(x_1, b_1)$ if $y_1$ is centered in $\levelgraph_j$).
		\end{enumerate}
		\item \label{proof:fitem5} For any key $[wt_2, x_2, y_2 ] \geq [wt, x, y]$, let
		$G'$ contain $c > 0$ number of LHPs of weight  ${ wt_2}$ of the
		form $(x_2a_2, b_2y_2)$. Further, let ${ c_{new}}$ (resp. ${ c_{old }}$) denote the number of such LHPs
		that are {\em new} (resp. not {\em new}).
		Here ${ c_{new} + c_{old} = c}$. Then the tuple $(x_2 a_2, b_2 y_2) \in$ \MTend, iff 
		$c_{old} > 0$ and $c_{new}$ paths have been added to $H_f$ by some earlier iteration of the while loop.
	\end{enumerate}
	
	Initialization and Maintenance for the invariant assertions above are similar to the proof of Lemma~\ref{lemma:fclean1}.
	
	{\noindent \bf Termination:} The condition to exit the loop is $H_f = \emptyset$. Because invariant assertion $\ref{proof:fitem1}$ maintains in $H_f$ the first triple to be extracted and processed, then $H_f = \emptyset$ implies that there are no triples, formed by a valid left or right extension, that contain \emph{new} SPs or LSPs, that need to be added or restored in the graph $G$. Moreover, because of invariant assertions $\ref{proof:fitem2}$ and $\ref{proof:fitem21}$, every triple containing the node $v$, extracted and processed before $H_f = \emptyset$, has been added or restored with its correct count in the \systemend. Finally, for invariant assertions $\ref{proof:fitem3}$ and $\ref{proof:flitem21}$, the sets $L,R$ and $L^{*}, LC^{*}, R^{*}, RC^{*}$ for each level, are correctly maintained. This completes the proof of the loop invariant.
	
	\vone
	By Lemma \ref{fdfixcorr}, all the SP distances in $G'$ are placed in $H_f$ and processed by the algorithm. Hence, after Algorithm~\ref{algo:fffixup} is executed, every SP in $G'$ is in its corresponding $P^*$ by the invariant of Lemma~\ref{proof:ffflem}.   
	Since every LST of the form $(xa,by)$ in $G'$ is formed by a left extension of a set of STs of the form $(a\times,by)$ (Step \ref{ffixup:lext} -  Algorithm~\ref{algo:fffixup}), or a right extension of a set of the form $(xa,\times b)$ (analogous steps for right extensions), and all the STs are correctly maintained and extendend (by the invariant of Lemma \ref{proof:ffflem}), then all the LSTs are correctly maintained at the end of \ffullyfixupend. This completes the proof of the Lemma.
	\hfill$\qed$
\end{proof}

\bibliographystyle{abbrv}
\bibliography{references,refs2}
\end{document}